\documentclass[11pt]{iopart}

\usepackage{graphicx}
\usepackage{latexsym}
\usepackage{textcomp}
\usepackage{iopams}
\usepackage{setstack}
\usepackage{multirow,booktabs}
\usepackage{dsfont}
\usepackage{url,changes}
\usepackage[utf8]{inputenc}
\usepackage[english]{babel}
\newcommand{\binom}[2]{\Bigl(\begin{array}{@{}c@{}}#1\\#2\end{array}\Bigr)}

\usepackage[normalem]{ulem}
\usepackage{multirow}
\usepackage{bbm,xcolor}
\usepackage{palatino}
\usepackage{mathpazo}

\usepackage{latexsym}
\usepackage{amsthm}
\usepackage{mathrsfs}
\usepackage{color,verbatim,graphics}
\usepackage{color,verbatim,graphics}
\usepackage[makeroom]{cancel}
\DeclareMathAlphabet{\mathrsfs}{U}{rsfs}{m}{n}
\DeclareMathAlphabet{\mathpzc}{OT1}{pzc}{m}{it}
\DeclareMathAlphabet{\matheus}{U}{eus}{m}{n}
\DeclareMathAlphabet{\mathbbold}{U}{bbold}{m}{n}

\definecolor{myurlcolor}{rgb}{0,0,0.7}
\definecolor{myrefcolor}{rgb}{0,0.7,0}
\usepackage[unicode=true,pdfusetitle, bookmarks=true,bookmarksnumbered=false,bookmarksopen=false, breaklinks=false,pdfborder={0 0 0},backref=false,colorlinks=true, linkcolor=myrefcolor, citecolor=myurlcolor,urlcolor=myurlcolor]{hyperref}

\newcommand{\ba}{\begin{eqnarray}}
\newcommand{\be}{\begin{equation}}
\newcommand{\ee}{\end{equation}}
\newcommand{\beq}{\begin{equation}}
\newcommand{\eeq}{  \end{equation}}
\newcommand{\bea}{\begin{eqnarray}}
\newcommand{\eea}{  \end{eqnarray}}

\newcommand{\ea}{\end{eqnarray}}
\newcommand{\ban}{\begin{eqnarray*}}
\newcommand{\ean}{\end{eqnarray*}}

\newcommand{\ket}[1]{|#1\rangle}
\newcommand{\bra}[1]{\langle#1|}

\newcommand{\bracket}[3]{\langle#1|#2|#3\rangle}

\newcommand{\ps}{\ket{\psi}}

\newcommand{\ct}{c_{\theta}}
\newcommand{\st}{s_{\theta}}
\newcommand{\Ct}{c^2_{\theta}}

\newcommand{\idd}{\mathbb{I}}

\newcommand{\Ah}{\hat{A}}
\newcommand{\Bh}{\hat{B}}
\newcommand{\Ch}{\hat{C}}

\newtheorem{theorem}{Theorem}
\newtheorem{cor}{Corollary}
\newtheorem{lem}{Lemma}
\newtheorem{defn}{Definition}

\newcommand{\ivan}{\color{black}}

\newcommand{\ot}[0]{\otimes}

\begin{document}

\title{Self-testing multipartite entangled states through projections onto two systems}

\author{I. \v{S}upi\'{c}$^1$, A. Coladangelo$^2$, R. Augusiak$^3$, A. Ac\'in$^{1,4}$}
\address{$^1$ ICFO-Institut de Ciencies Fotoniques, The Barcelona Institute of Science and Technology, 08860 Castelldefels (Barcelona), Spain}
\address{$^2$  Department of Computing and Mathematical Sciences, California Institute of Technology,
1200 E California Blvd, Pasadena, CA 91125, United States}
\address{$^3$ Center for Theoretical Physics, Polish Academy of Sciences, Aleja Lotnik\'{o}w 32/46, 02-668 Warsaw, Poland}
\address{$^4$ ICREA--Institucio Catalana de Recerca i Estudis Avan\c{c}ats, Lluis Companys 23, 08010 Barcelona, Spain}
\ead{ivan.supic@icfo.es}

\date{\today}

\begin{abstract}
Finding ways to test the behaviour of quantum devices is a timely enterprise, especially in the light of the rapid development of quantum technologies. Device-independent self-testing is one desirable approach, as it makes minimal assumptions on the devices being tested. In this work, we address the question of which states can be self-tested. This has been answered recently in the bipartite case \cite{Coladangelo}, while it is largely unexplored in the multipartite case, with only a few scattered results, using a variety of different methods: maximal violation of a Bell inequality, numerical SWAP method, stabilizer self-testing etc. In this work, we investigate a simple, and potentially unifying, approach: combining projections onto two-qubit spaces (projecting parties or degrees of freedom) and then using maximal violation of the tilted CHSH inequalities. This allows to obtain self-testing of Dicke states and partially entangled GHZ states with two measurements per party, and also to recover self-testing of graph states (previously known only through stabilizer methods). Finally, we give the first self-test of a class multipartite qudit states: we generalize the self-testing of partially entangled GHZ states by adapting techniques from \cite{Coladangelo}, and show that all multipartite states which admit a Schmidt decomposition can be self-tested with few measurements. 
\end{abstract}

\section{Introduction}

The rapid development of quantum technologies in recent years creates an urgent need for certification tools. Quantum computing and quantum simulation are state of the art tasks which require verifiable realizations. One way to certify the correct functioning of a quantum computer would be to ask it to solve a problem that is thought to be hard for a classical computer, like factoring large numbers and simply checking the correctness of the solution. However, it is conjectured that the class of problems that can be solved efficiently on a quantum computer (BQP) has elements outside the class of problems whose solution can be checked classically (NP) \cite{Aaronson}, which makes this type of verification incomplete. Thus, efforts have been made towards building reliable certification protocols for quantum systems performing universal quantum computing or quantum simulations \cite{UBQC,AGKE, RUV}. In this work, we investigate one of the basic building blocks for such verifications tasks, namely certification of a particular quantum state.

A canonical way to approach {\ivan{the problem of certification of quantum states}} is to exploit tomographic protocols \cite{bible}. Unfortunately, quantum devices performing tasks such as quantum computation typically involve multipartite quantum states and the complexity of tomographic techniques scales exponentially with the number of particles involved. Moreover, they demand a set of trusted measurements, which in certain scenarios is not an available resource.

An alternative technique able to positively address these problems is \textit{self-testing} \cite{MY}. Contrary to quantum state and process tomography, self-testing is a completely device-independent task. It aims to verify that a given quantum device operates on a certain quantum state, and performs certain measurements on it, solely from the correlations it generates.
The building block for this, as well as for all other device-independent protocols is Bell's theorem \cite{Bell}, which says that correlations violating Bell inequalities do not admit local hidden-variable models. Thus, correlations useful for self-testing must be non-local. Self-testing was formally introduced by Mayers and Yao \cite{MY}. Since then, there has been growing interest in designing self-testing methods \cite{MYS,Wang,Cedric,SASA,jed2}, and studying their robustness \cite{MYS,swap,jed}. An important recent development shows that all pure entangled bipartite states can be self-tested \cite{Coladangelo}. {\ivan {As for applications, self-testing methods have successfully been employed in protocols for quantum key distribution \cite{VV14, Miller14}, randomness expansion \cite{Miller14}, and verification of quantum computations \cite{RUV,leash,graph,GWK}}}.

Most of the currently known self-testing protocols, however, are tailored to bipartite states, leaving the multipartite scenario rather unexplored. The known examples cover only the tripartite $W$ state, a class of partially entangled tripartite states \cite{Singapur,Pal} and graph states \cite{McKague}.
The aim of this paper is to extend the class of multipartite states that can be self-tested, by investigating a simple approach that exploits the well-understood self-testing of two-qubit states. At a high level, this is done by combining projections to two-qubit spaces and then exploiting maximal violation of tilted CHSH inequalities. Using this potentially unifying approach, we show self-testing of all Dicke states and partially entangled GHZ states with only two measurements per party. We also show that our method efficiently applies also to self-testing of graph states, previously known only through stabilizer state methods, with a slight improvement in the number of measurement settings per party. Finally, using techniques from \cite{Coladangelo} as a building block, we provide the first self-testing result for a class of multipartite qudit states, by showing that all multipartite qudit states which possess a Schmidt decomposition can be self-tested, with at most four measurements per party\footnote{{\ivan{By "$n$ measurements per party" we mean that each party has $n$ different measurement settings, not that the total number of measurement rounds performed in the experiment  by each party is equal to $n$.}}}.\\


\section{Preliminaries}\label{sec:prel}

Self-testing is a device-independent task \cite{DI} whose aim is to characterize the form of the quantum state and measurements solely from the correlations that they generate. To introduce it formally, consider $N$ non-communicating parties sharing some $N$-partite state $\ket{\psi}$. On its share of this state, party $i$ can perform one of several projective measurements  $ \{M^{a_i}_{x_i,i}\}_{a_i}$, labelled by  $x_i \in \mathcal{X}_i$, with possible outcomes $a_i \in \mathcal{A}_i$. Here $\mathcal{X}_i$ and $\mathcal{A}_i$ stand for finite alphabets of possible questions and answers for party $i$.
The experiment is characterized by a collection of conditional probabilities %
%
$\{p(a_1,\ldots,a_N|x_1,\ldots,x_N): a_i \in \mathcal{A}_i\}_{x_i \in \mathcal{X}_i}$, where 
\begin{equation}\label{eq:corrVec}
p(a_1,\ldots,a_N|x_1,\ldots,x_N)
=\langle\psi|M_{x_1,1}^{a_1}\otimes\ldots\otimes M_{x_N,N}^{a_N}|\psi \rangle
\end{equation}
is the probability of obtaining outputs $a_1,\ldots,a_N$ upon performing the measurements $x_1,\ldots,x_N$\footnote{We take the parties' measurements to be projective, invoking Naimark's dilation theorem. We take the joint state to be pure for ease of exposition, but we emphasize that all of our proofs hold analogously starting from a joint mixed state.}. We refer to this as a \textit{correlation}. It is sometimes convenient to describe correlations with the aid of standard correlators, where instead of measurement operators $M_{x_i}^{a_i}$ one uses Hermitian observables with eigenvalues $\pm 1$. Now, we can formally define self-testing in the following way.

\begin{defn}[Self-testing]
We say that a correlation $\{p(a_1,\ldots,a_N|x_1,\ldots,x_N): a_i \in \mathcal{A}_i\}_{x_i \in \mathcal{X}_i}$ self-tests the state $\ket{\psi'}$ and measurements $\{{\tilde{M}}_{x_i,i}^{a_i}\}_{a_i}$, $i=1,\ldots,N$, if for any state and measurements $\ket{\psi}$ and $\{M_{x_i,i}^{a_i}\}_{a_i}$, $i=1,\ldots,N$, reproducing the correlation, there exists a local isometry $\Phi= \Phi_1\otimes\ldots\otimes \Phi_N$ such that 
\begin{equation}\label{selftest}
\Phi(M_{x_1,1}^{a_1}\otimes\ldots\otimes M_{x_N,N}^{a_N}\ket{\psi})=\ket{\mathrm{aux}}\otimes ({\tilde{M}}_{x_1,1}^{a_1}\otimes\ldots\otimes {\tilde{M}}_{x_N,N}^{a_N}\ket{\psi'}).
\end{equation}
where $\ket{\mathrm{aux}}$ is an auxiliary state.
\end{defn}

{\ivan{Intuitively, we can think of the self-testing correlations as characterizing the state and measurements that achieve them, up to a local isometry.}} {\ivan{Indeed, this is the best possible characterization that one can hope for. In fact, one can never characterize the state and measurements exactly by the observed correlation: this is because correlations are invariant under local unitary transformations or embeddings in Hilbert spaces of higher dimension.}} {\ivan{Moreover, we point out that it is not possible to self-test mixed states: this is because for any correlation that can be obtained by measuring a mixed state, there exists a pure state of the same dimension that can be measured to obtain the same correlations \cite{Sikora16}. }} 

In some cases the existence of the required isometry can be proven solely from the maximal violation of some Bell inequality. For instance, all two-qubit pure entangled states can be self-tested with a one-parameter class of tilted CHSH Bell inequalities \cite{Cedric} given by
\begin{equation}\label{tCHSH}
\alpha\langle A_0\rangle+\langle A_0B_0\rangle+\langle A_0B_1\rangle+\langle A_1B_0\rangle-\langle A_1B_1\rangle\leq 2+\alpha,
\end{equation}
where $\alpha\geq 0$ and $A_i$ and $B_i$ are  observables with outcomes $\pm 1$ measured by the parties. Note that for $\alpha=0$, (\ref{tCHSH}) reproduces the well-known CHSH Bell inequality \cite{CHSH}. For further purposes let us briefly recall this result. Here $\sigma_z$ and $\sigma_x$ are the standard Pauli matrices.

\begin{lem}[\cite{Cedric}]\label{LemmaTilt}
Suppose a bipartite state $\ket{\psi}$ and dichotomic observables $A_i$ and $B_i$ achieve the maximal quantum violation of the tilted CHSH inequality (\ref{tCHSH}) $\sqrt{8+2\alpha^2}$, for some $\alpha$. 
Let $\theta, \mu \in (0,\pi/2)$ be such that $\sin 2\theta = [(4-\alpha^2)/(4+\alpha^2)]^{\frac{1}{2}}$ and $\mu = \arctan \sin 2\theta$. Let $Z_A = A_0$, $X_A = A_1$. Let $Z^*_B$ and $X^*_B$ be respectively $(B_0 + B_1)/2\cos \mu$ and $(B_0 - B_1)/2\sin \mu$, but with all zero eigenvalues replaced by one, and define $Z_B = Z^*_B|Z^*_B|^{-1}$ and $X_B = X^*_B|X^*_B|^{-1}$. Then, we have
\begin{eqnarray}
Z_A\ps = Z_B \ps, \label{eq4} \\
\cos\theta X_A (\mathds{1}-Z_A) \ps =  \sin\theta X_B (\mathds{1}+Z_A) \ps, \label{eq5}\\
\{Z_A,X_A\}\ps = 0, \qquad \{Z_B,X_B\}\ps = 0. \nonumber 
\end{eqnarray}
Moreover, there exists a local isometry $\Phi$ such that $\Phi(A_i\otimes B_j\ket{\psi})=\ket{\mathrm{aux}}\otimes (\tilde{A}_i\otimes \tilde{B}_j)\ket{\psi_{\theta}}$, where $\ket{\psi_{\theta}}=\cos\theta\ket{00}+\sin\theta\ket{11}$, and $\tilde{A}_0=\sigma_z$, $\tilde{A}_1=\sigma_x$, and $\tilde{B}_{0/1}=\cos\mu\sigma_{z}\pm\sin\mu\sigma_x$.
\end{lem}

A typical construction of the isometry  $\Phi$ is the one encoding the SWAP gate, as illustrated in Figure \ref{fig:method}. 
\begin{figure}[h!]
\centering
\includegraphics[width=0.45\columnwidth]{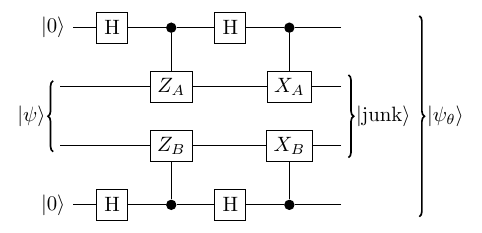}
\caption{Example of a circuit that takes as input a state $\ket{\psi}$ satisfying \ref{eq4}-\ref{eq5}, adds two ancillas, each in $\ket{0}$, and outputs the state $\ket{\psi_{\theta}}$ in tensor product with an auxiliary state $\ket{\mathrm{aux}}$. Here $H$ is the usual Hadamard gate. \label{fig:method}}
\end{figure}
 
Our aim in this paper is to exploit the above result to develop methods for self-testing multipartite entangled quantum states. Given an $N$-partite entangled state $\ket{\psi}$, the idea is that $N-2$ chosen parties perform local measurements on their shares of $\ket{\psi}$ and the remaining two parties check whether the projected state they share violates maximally (\ref{tCHSH}) for the appropriate $\alpha$ (we can think of this as a sub-test). This procedure is repeated for various subsets of $N-2$ parties until the correlations imposed are sufficient to characterize the state $\ket{\psi}$. Our approach is inspired by Ref. \cite{Singapur}, which shows that any state in the class $(\ket{100}+\ket{101}+\alpha\ket{001})/\sqrt{2+\alpha^2}$, containing the three-qubit $W$ state, can be self-tested in this way. We will show that this approach can be generalized in order to self-test new (and old) classes of multipartite states. The main challenge is to show that all the sub-tests of different pairs of parties are compatible. To be more precise, for a generic state there will always be a party which will be involved in several different sub-tests and, in principle, will be required to use different measurements to pass the different tests. Consequently, isometries (Figure \ref{fig:method}) corresponding to different sub-tests are in principle constructed from different observables. However, a single isometry is required in order to self-test the global state. Overcoming the problem of building a single isometry from several different ones is the key step to achieve a valid self-test for multipartite states. For states that exhibit certain symmetries, this can be done efficiently with few measurements. We leave for future work the exploration for states that do not have any particular symmetry.

In the $N$-partite scenario, parties will be denoted by numbers from $1$ to $N$ and measurement observables by capital letters with a superscript denoting the party. For a two-outcome observable $W$, we denote by $W^{(\pm)} = (\idd \pm W)/2$ the projectors onto the $\pm 1$ eigenspaces. We use the notation $\lfloor a \rfloor$ to denote the biggest integer $n$ such that $n \leq a$, while $\lceil a \rceil$ is the smallest $n$ such that $n \geq a$. \\
 
\section{Results}\label{sec:res}

In this work, we expand the class of self-testable multipartite states. More precisely, in subsection \ref{subsection_schmidt} we show that all multipartite partially entangled GHZ (qubit) states can be self-tested with two measurements per party. Then, we make use of this result as a building block to extend self-testing to all multipartite entangled Schmidt-decomposable qudit states, of any local dimension $d$, with only three measurements per party (except one party has four). To the best of our knowledge, this is the first self-test for multipartite states of qudits, for $d>2$. Finally, in subsections \ref{subsec:Dicke} and \ref{subsec:graph} we apply the approach used for multipartite partially entangled GHZ (qubit) states to the self-test the classes of Dicke states and graph states (previously known to be self-testable through stabilizer methods \cite{McKague}).
\subsection{All multipartite entangled qudit Schmidt states}
\label{subsection_schmidt}

While in the bipartite setting all states admit a Schmidt decomposition, in the general multipartite setting this is not the case. We refer to those multipartite states that admit a Schmidt decomposition as Schmidt states. By a Schmidt state we consider a multipartite state $\ket{\Psi}$ acting on $\mathbb{C}^{d_1}\otimes\mathbb{C}^{d_2} \cdots \otimes\mathbb{C}^{d_N}$ that by local unitary transformations can be brought to the following form
\begin{equation}\label{SchmidtState}
\ket{\Psi} = \sum_{j=0}^{d-1} c_j \ket{j}^{\otimes N}
\end{equation}
where $d = \min_i d_i$, $0< c_j < 1$ for all $i$ and $\sum_{j=0}^{d-1}c_j^2 = 1$. Typicality of multipartite states that are equivalent to some Schmidt state up to a local unitary is discussed in \cite{Peres,Thapliyal}.

Our proof that all multipartite entangled Schmidt states can be self-tested follows closely the ideas from \cite{Coladangelo}, while leveraging as a building block our novel self-testing result for partially entangled GHZ states. Thus, we proceed by first proving a self-testing theorem for multipartite partially entangled qubit GHZ states.\\

\textit{\textbf{Multipartite partially entangled GHZ states.}} Multipartite qubit Schmidt states, also known as partially entangled GHZ states, are of the form
\begin{equation}
\ket{\mathrm{GHZ}_N(\theta)}=\cos\theta\ket{0}^{\otimes N}+
\sin\theta\ket{1}^{\otimes N}
\end{equation}
where $\theta\in(0,\pi/4]$ and $\ket{\mathrm{GHZ}_N(\pi/4)}=\ket{\mathrm{GHZ}_N}$ is the standard $N$-qubit GHZ state. The form of this state is such that if any subset of $N-2$ parties performs a 
$\sigma_X$ measurement, the collapsed state
shared by the remaining two parties is $\cos\theta\ket{00}\pm\sin\theta\ket{11}$, depending on the parity of the measurement outcomes. As already mentioned, these states can be self-tested with the aid of inequality (\ref{tCHSH}), which is the main ingredient of our self-test of $\ket{\mathrm{GHZ}_N(\theta)}$. 

\begin{theorem}
\label{thm-PEGHZ}
Let $\ket{\psi}$ be an $N$-partite state, and let $A_{0,i},A_{1,i}$ be a pair of binary observables for the $i$-th party, for $i=1,\ldots,N$. Suppose the following correlations are satisfied:
\begin{eqnarray} \label{prvap}
\fl \bra{\psi}A_{0,i}^{(+)}\ket{\psi} = \bra{\psi}A_{0,i}^{(+)}A_{0,j}^{(+)}\ket{\psi} = \cos^2\theta ,\qquad \forall i,j \in \{1,\ldots,N-1\}\\ \label{drugad}
\fl \bra{\psi}\prod_{i=1}^{N-2} A_{1,i}^{(a_i)}\ket{\psi} = \frac{1}{2^{N-2}},\qquad \forall a \in \{+,-\}^{N-2} \\ \label{trecat}
\fl \bra{\psi}\prod_{i=1}^{N-2} A_{1,i}^{(a_i)}(\alpha A_{0,N-1} + A_{0,N-1} A_{0,N} + A_{0,N-1} A_{1,N}+ (-1)^{h(a)} A_{1,N-1} A_{0,N} \\ 
\fl \hspace{2cm}- (-1)^{h(a)} A_{1,N-2}A_{1,N-1} )\ket{\psi} 
= \frac{\sqrt{8+2\alpha^2}}{2^{N-2}},\qquad \forall a \in \{+,-\}^{N-2} 
\end{eqnarray}
where $h(a)$ denotes the parity of the number of ``$-$'' in $a$, and $\alpha=2\cos2\theta/\sqrt{1+\sin^22\theta}$. Let $\mu$ be such that $\tan{\mu} = \sin{2\theta}$. Define $Z_i = A_{0,i}$ and $X_i = A_{1,i}$, for $i=1,\ldots,N-1$. Then, let $Z'_N = (A_{0,N} + A_{1,N})/2\cos{\mu}$, and let $Z^*_N $ be $Z'_N$ with zero eigenvalues replaced by $1$. Define $Z_N = Z^*_N |Z^*_N  |^{-1}$. Define $X_N$ similarly starting from $X'_N = (A_{0,N} - A_{1,N})/2\sin{\mu}$. Then,
\begin{eqnarray}
Z_1 \ps=  \cdots = Z_N \ps, \label{eq12}\\
X_1\cdots X_D(I - Z_1)\ps  = \tan{\theta} (I + Z_1)\ps. \label{eq13}
\end{eqnarray}
\end{theorem}

\noindent \textit{Proof:} We refer the reader to \ref{appendix_C} for the formal proof of this Theorem, while providing here an intuitive understanding of the correlations given above. The first equation (\ref{prvap}) defines the existence of one measurement observable, whose marginal carries the information of angle $\theta$. The straightforward consequence of it is  (\ref{eq12}), which is analogue to  (\ref{eq4}). On the other hand, (\ref{drugad}) involves a different measurement observable with zero marginal, 
 while  (\ref{trecat}) shows that when the first $N-2$ parties perform this zero marginal measurement the remaining two parties maximally violate the corresponding tilted CHSH inequality, i.e. the reduced state is self-tested to be the partially entangled pair of qubits. Note that (\ref{eq13}) is analogue to (\ref{eq5}).  \\

As a corollary, these correlations self-test the state $\ket{\mathrm{GHZ}_N(\theta)}$.

\begin{cor}
Let $\ket{\psi}$ be an $N$-partite state, and let $A_{0,i},A_{1,i}$ be a pair of binary observables for the $i$th party, for $i=1,\ldots,N$. Suppose they satisfy the correlations of Theorem \ref{thm-PEGHZ}. Then, there exists a local isometry $\Phi$ such that
\begin{equation}
\Phi(\ket{\psi}) = \ket{\mathrm{aux}}\ket{\mathrm{GHZ}_N(\theta)}
\end{equation}
\end{cor}

\noindent \textit{Proof:} This follows as a special case ($d=2$) of Lemma \ref{yngen} stated below, upon defining $P_i^{(k)} =[I+ (-1)^{k}Z_i]/2$, for $k \in \{0,1\}$.

As one can expect, the ideal measurements achieving these correlations are: $A_{0,i} = \sigma_z$, $A_{1,i} = \sigma_x$, for $i=1,\ldots,N-1$, and $A_{0,N} = \cos{\theta} \sigma_z + \sin{\theta} \sigma_x$, $A_{1,N} =  \cos{\theta} \sigma_z - \sin{\theta} \sigma_x$. We refer to the correlations achieved by these ideal measurements as the \textit{ideal correlations} for multipartite entangled GHZ states. \\

\textit{\textbf{All multipartite entangled qudit Schmidt states.}} The generalisation of Theorem \ref{thm-PEGHZ} to all multipartite qudit Schmidt states is then an adaptation of the proof in \cite{Coladangelo} for the bipartite case, with the difference that it uses as a building block the $\ket{\mathrm{GHZ}_N(\theta)}$ self-test that we just developed, instead of the tilted CHSH inequality. 

We begin by stating a straightforward generalisation to the multipartite setting of the criterion from \cite{YN} which gives sufficient conditions for self-testing a Schmidt state. Then, our proof that all multipartite entangled qudit Schmidt states can be self-tested goes through showing the existence of operators satisfying the conditions of such criterion. 

\begin{lem}[Generalisation of criterion from \cite{YN}]
\label{yngen} Let $\ket{\Psi}$ be a state of the form (\ref{SchmidtState}). Suppose there exist sets of unitaries $\{X^{(k)}_{l}\}_{k=0}^{d-1}$, where the subscript $l \in \{1,\dots , N\}$ indicates that the operator acts on the system of the $l$-th party, and sets of projections $\{P^{(k)}_{l}\}_{k =0}^{d-1}$, that are complete and orthogonal for $l=1,\ldots,N-1$ and need not be such for $l=N$, and they satisfy:
\begin{eqnarray}
P^{(k)}_{1}\ket{\psi}=\dots = P^{(k)}_{N}\ket{\psi}, \label{c2}\\
X^{(k)}_{1}\dots X^{(k)}_{N}P^{(k)}_{1}\ket{\psi} = \frac{c_k}{c_0} P^{(0)}_{1}\ket{\psi}\label{c3}
\end{eqnarray}
for all $k=1,\ldots,N$. Then, there exists a local isometry $\Phi$ such that $\Phi(\ket{\psi})=\ket{aux}\otimes\ket{\Psi}$.
\end{lem}
\proof{The proof of Lemma \ref{yngen} is a straightforward generalisation of the proof of the criterion from \cite{YN}, and is included in the Appendix for completeness.}

We now describe the self-testing correlations for $\ket{\Psi} = \sum_{j=0}^{d-1} c_j \ket{j}^{\otimes n}$. Their structure is inspired by the self-testing correlations from \cite{Coladangelo} for the bipartite case, and they consist of three $d$-outcome measurements for all but the last party, which has four. We desribe them by first presenting the ideal measurements that achieve them, as we believe this aids understading. Subsequently, we extract their essential properties that guarantee self-testing. 
For a single-qubit observable $A$, denote by $[A]_m$ the observable defined with respect to the basis $\{\ket{2m~\textrm{mod}~d},\ket{(2m+1)~\textrm{mod}~d}\}$. For example, $[\sigma_Z]_m = \ket{2m}\bra{2m} - \ket{2m+1}\bra{2m+1}$. Similarly, we denote by $[A]'_m$ the observable defined with respect to the basis $\{\ket{(2m+1)~\textrm{mod}~d}, \ket{(2m+2)~\textrm{mod}~ d}\}$. We use the notation $\bigoplus A_i$ to denote the direct sum of observables $A_i$.

Let $\mathcal{X}_i$ denote the question alphabet of the $i$-th party, and let $\mathcal{X}_i = \{0,1,2\}$ for $i=1,\ldots,N-1$, and $\mathcal{X}_N = \{0,1,2,3\}$. Let $x_i \in \mathcal{X}_i$ denote a question to the $i$-th party.  The answer alphabets are $\mathcal{A}_i = \{0,1,\ldots,d-1\}$, for $i=1,\ldots,N$.

\begin{defn}[Ideal measurements for multipartite entangled Schmidt states]
\label{id_meas_schmidt}
The $N$ parties make the following measurements on the joint state $\ket{\Psi} = \sum_{j=0}^{d-1} c_j \ket{j}^{\otimes n}$. \\

\noindent\textbf{For $i=1,\ldots,N-1$:}

\begin{itemize}

\item For question $x_i = 0$, the $i$-th party measures in the computational basis $\{\ket{0},\ket{1},\cdots,\ket{d-1}\}$ of its system,
 
\item For $x_i=1$ and $x_i=2$: for $d$ even, in the eigenbases of observables $\bigoplus_{m=0}^{\frac{d}{2}-1} [\sigma_X]_m$ and $\bigoplus_{m=0}^{\frac{d}{2}-1} [\sigma_X]'_m$ respectively, with the natural assignments of $d$ measurement outcomes; for $d$ odd, in the eigenbases of observables $\bigoplus_{m=0}^{\frac{d-1}{2}-1} [\sigma_X]_m\oplus \ket{d-1}\bra{d-1}$ and $\ket{0}\bra{0}\oplus\bigoplus_{m=0}^{\frac{d-1}{2}-1} [\sigma_X]'_m$ respectively.
 
 \end{itemize}

\noindent\textbf{For $i=N$:}
\begin{itemize}

\item For $x_N = 0$ and $x_N = 1$, the party $N$ measures in the eigenbases of $\bigoplus_{m=0}^{\frac{d}{2}-1} [\cos{(\mu_m)}\sigma_Z+\sin{(\mu_m)}\sigma_X]_m$ and $\bigoplus_{m=0}^{\frac{d}{2}-1} [\cos{(\mu_m)}\sigma_Z-\sin{(\mu_m)}\sigma_X]_m$ respectively, with the natural assignments of $d$ measurement outcomes, where $\mu_m = \arctan[\sin(2\theta_m)]$ and $\theta_m  = \arctan(c_{2m+1}/c_{2m})$; for $d$ odd, he measures in the eigenbases of $\bigoplus_{m=0}^{\frac{d-1}{2}-1} [\cos{(\mu_m)}\sigma_Z+\sin{(\mu_m)}\sigma_X]_m\oplus\ket{d-1}\bra{d-1}$ and $\bigoplus_{m=0}^{\frac{d-1}{2}-1} [\cos{(\mu_m)}\sigma_Z-\sin{(\mu_m)}\sigma_X]_m\oplus\ket{d-1}\bra{d-1}$ respectively. 

\item For $x_N=2$ and $x_N=3$: for $d$ even, the $N$-th party measures in the eigenbases of $\bigoplus_{m=0}^{\frac{d}{2}-1} [\cos{(\mu'_m)}\sigma_Z+\sin{(\mu'_m)}\sigma_X]'_m$ and $\bigoplus_{m=0}^{\frac{d}{2}-1} [\cos{(\mu'_m)}\sigma_Z-\sin{(\mu'_m)}\sigma_X]'_m$ respectively, where $\mu'_m = \arctan[\sin(2\theta'_m)]$ and $\theta'_m = \arctan(c_{2m+2}/c_{2m+1})$; for $d$ odd, in the eigenbases of $\ket{0}\bra{0}\oplus\bigoplus_{m=0}^{\frac{d-1}{2}-1} [\cos{(\mu'_m)}\sigma_Z+\sin{(\mu'_m)}\sigma_X]'_m$ and $\ket{0}\bra{0}\oplus\bigoplus_{m=0}^{\frac{d-1}{2}-1} [\cos{(\mu'_m)}\sigma_Z-\sin{(\mu'_m)}\sigma_X]'_m$, respectively.

\end{itemize}

We refer to the correlation specified by the ideal measurements above as the \textit{ideal correlation} for multipartite entangled Schmidt states.

\end{defn}
Next, we will highlight a set of properties of the ideal correlation that are enough to characterize it, in the sense that any quantum correlation that satisfies these properties has to be the ideal one. This also aids understanding of the self-testing proof (Proof of Theorem \ref{thm_schmidt}). In what follows, we will employ the language of correlation tables, which gives a convenient way to describe correlations. In general, let $\mathcal{X}_i$ be the question alphabers and $\mathcal{A}_i$ the answer alphabets. A correlation specifies, for each possible question $x \in \mathcal{X}_1 \times \cdots \times \mathcal{X_N}$, a table $T_{x}$ with entries $T_x(a) = p(a|x)$ for $a \in \mathcal{A}_1 \times \cdots \times \mathcal{A}_N$. For example, we denote the correlation tables for the ideal correlations for multipartite entangled GHZ states from Theorem \ref{thm-PEGHZ} as $T_{x}^{\mbox{ghz}_N(\theta_m)}$, where $x \in \{0,1\}^N$ denotes the question. 

\begin{defn}[Self-testing properties of the ideal correlations for multipartite entangled Schmidt states]
Recall that $\mathcal{X}_i = \{0,1,2\}$ for $i=1,\ldots,N-1$, and $\mathcal{X}_N = \{0,1,2,3\}$. $\mathcal{A}_i = \{0,1,\ldots,d-1\}$, for $i=1,\ldots,N$. 

The self-testing properties of the ideal correlations are:
\begin{itemize}
\item For questions $x \in \{0,1\}^N$, we require $T_{x}$ to be block-diagonal with $2^{\times N}$ blocks $C_{x,m} :=(c_{2m}^2+c_{2m+1}^2) \cdot T_{x}^{\mbox{ghz}_N(\theta_m)}$ corresponding to outcomes in $\{2m,2m+1\}^N$, where the multiplication by the weight is intended entry-wise, and $\theta_m := \arctan \big(c_{2m+1}/c_{2m}\big)$. 
\item For questions with $x_i \in\{0,2\}$, for $i=1,\ldots,N-1$ and $x_N \in \{2,3\}$ we require $T_{x}$ to be block-diagonal with the $2^{\times N}$ blocks "shifted down" by one measurement outcome. These should be $D_{x,m} := (c_{2m+1}^2+c_{2m+2}^2) \cdot T_{f(x_1),\ldots,f(x_{N-1}),g(x_N)}^{\mbox{ghz}_N(\theta_m')}$ corresponding to measurement outcomes in $\{2m+1,2m+2\}^N$, where $\theta'_m := \arctan \big(c_{2m+2}/c_{2m+1}\big)$ and $f(0) = 0$, $f(2) = 1$, $g(2)= 0$, $g(3)=1$.
\end{itemize}
\end{defn}

We are now ready to state the main theorem of this section. 

\begin{theorem}
\label{thm_schmidt}
Let $\ket{\Psi} = \sum_{j=0}^{d-1} c_j \ket{j}^{\otimes N}$, where $0< c_j < 1$ for all $i$ and $\sum_{j=0}^{d-1}c_j^2 = 1$. Suppose $N$ parties exhibit the ideal correlations for multipartite entangled Schmidt states from Definition \ref{id_meas_schmidt} by making local measurements on a joint state $\ket{\psi}$. Then there exists a local isometry $\Phi$ such that $\Phi(\ket{\psi}) = \ket{aux} \otimes \ket{\Psi}$.
\end{theorem}
As we mentioned, the proof of Theorem \ref{thm_schmidt} follows closely the method of \cite{Coladangelo}, and uses as a building block our self-testing of the $n$-partite partially entangled GHZ state. For the details, we refer the reader to \ref{appendix_D}.

\subsection{Symmetric Dicke states}
\label{subsec:Dicke}

Let us now consider the symmetric Dicke states. These are simultaneous eigenstates of the square of the total angular momentum operator $\bi{J}^2$
of $N$ qubits and its projection onto the $z$-axis $J_z$. In a concise way they can be stated as
\be\label{eq:dicke1}
|D_N^k\rangle = {{N}\choose{k}}^{-\frac{1}{2}}\sum_i P_i(\ket{1}^{\otimes k} \ket{0}^{\otimes(N-k)}),
\ee
where the sum goes over all permutations of the parties and $k$ is the number of excitations.
For instance, for $k=1$ they reproduce the $N$-qubit $W$ state:
\begin{equation}
|W_N\rangle=\frac{1}{\sqrt{N}}(\ket{0\ldots 01}+\ket{0\ldots 10}+\ldots+\ket{10\ldots 0}).
\end{equation}
Interestingly, Dicke states have been generated experimentally \cite{Vienna} and have important role in metrology tasks \cite{MetroDicke1} and quantum networking protocols \cite{Chiuri}. 

We now show how to self-test Dicke states. For convenience, we 
consider the unitarily equivalent state $|xD_N^k\rangle=\sigma_x^{N}|D_N^k\rangle$ with 
$\sigma_x^N$ denoting $\sigma_x$ applied to party $N$. Our self-test
exploits the fact that every Dicke state can be written as
\begin{equation}
|xD_{N}^k\rangle=
\frac{1}{\sqrt{N}}\left(\sqrt{N-m}\,\ket{0}|xD_{N-1}^k\rangle
+\sqrt{m}\,\ket{1}|xD_{N-1}^{k-1}\rangle\right)
\end{equation} 
which, after recursive application, allows one to express it in terms of the $(k+1)$-partite $W$ state, that is,
\begin{equation}\label{Dickedec}
|xD_{N}^k\rangle=\sum_{i_1,\ldots,i_{N-k-1}=0}^1\frac{{{k+1 \choose k-\Omega}}^{\frac{1}{2}}}{{{N\choose k}}^{\frac{1}{2}}}\ket{i_1,\ldots,i_{N-k-1}}|xD_{k+1}^{k-\Omega}\rangle,
\end{equation}
where the first ket is shared by the parties $1,\ldots,N-k-1$ and $\Omega=i_1+\ldots+i_{N-k-1}$. Now, for $i_1=\ldots=i_{N-k-1}=0$, the corresponding state $|xD_{k+1}^k\rangle$ is simply a rotated $(k+1)$-partite $W$ state $\sigma_x^{\otimes (k+1)}|xW_{k+1}\rangle$. Moreover, due to the fact that the Dicke states are symmetric, the above decomposition holds for any choice of $N-k-1$ parties among the first $N-1$ parties. Thus, if we had a self-test for the $N$-partite $W$ state $|xW_N\rangle$, we could use the above formula to generalize it to any Dicke state. Let us then show how to self-test any $W$ state.

\begin{theorem}\label{Thmw}
Let the state $\ket{\psi}$ and measurements
$Z_i,X_i$ for parties $i=1,\ldots,N-1$ and $D_N$ and $E_N$ for the last party, satisfy the following conditions:
\begin{equation}\label{Thmwcond1}
\left\langle \bigotimes_{l=1,l\neq i}^{N-1}Z_l^{(+)}\right\rangle=\frac{2}{N},\qquad
\left\langle \bigotimes_{l=1,l\neq i}^{N-1}Z_l^{(+)}\otimes B_{i,N}^{(+)}\right\rangle=\frac{4\sqrt{2}}{N},
\end{equation}
with $i=1,\ldots,N-1$, where, as before, $B_{i,N}^{(+)}=Z_i \otimes D_N+Z_i\otimes E_N+X_i\otimes D_N-X_i\otimes E_N$ is the Bell operator between the parties 
$i$ and $N$.
%
Moreover, we assume that 
\begin{equation}\label{Thmwcond2}
\langle Z_i^{(-)}\rangle=\frac{1}{N},\qquad
\left\langle \bigotimes_{l=1,l\neq i}^{N-1}Z_i^{(+)}\otimes Z_i^{(-)}\right\rangle=\frac{1}{N}
\end{equation}
with $i=1,\ldots,N-1$.
Then, for the isometry $\Phi_N$ one has
%
$\Phi_N(\ket{\psi})=\ket{\mathrm{aux}}\ket{xW_N}.$
\end{theorem}
We defer the detailed proof to \ref{App:W}, presenting here only a sketch. The proof makes use of the fact that $\ket{xW_N}$ can be written as $[\ket{0}^{\otimes N-2} (\ket{00}+\ket{11})_{i,N}+\ket{\mathrm{rest}_i}]/\sqrt{N}$, where $(\ket{00}+\ket{11})_{i,N}$ is the maximally entangled state between the parties $i$ and $N$, and the state $\ket{\mathrm{rest}_i}$ collects all the remaining kets. We thus impose in  (\ref{Thmwcond1}) that if $(N-2)$-partite subset of the first 
$N-1$ parties obtains $+1$ when measuring $Z_i$ on $\ket{\psi}$, the state held by the parties $i$ and $N$ violates maximally the CHSH Bell inequality. Conditions in (\ref{Thmwcond2}) are needed to characterize $\ket{\mathrm{rest}_i}$, which completes the proof.

Let us now demonstrate how the above result can be applied to self-test any Dicke state. First, let us simplify our considerations by noting that a Dicke state with $k\leq \lfloor N/2\rfloor$ is unitarily equivalent to a Dicke state with $m\geq \lceil N/2\rceil$, i.e., $|D_{N}^k\rangle=\sigma_z^{\otimes N}|D_{N}^{N-k}\rangle$ for $k=0,\ldots,\lfloor N/2\rfloor$. Thus, it is enough to consider the Dicke states with $k\geq \lfloor N/2\rfloor$. Second, due to the fact that Theorem \ref{Thmw} is formulated for $|xW_N\rangle$, while in the decomposition (\ref{Dickedec}) we have $\sigma_{x}^{\otimes (k+1)}|xW_N\rangle$, one has to modify the conditions in (\ref{Thmwcond1}) and (\ref{Thmwcond2}) as $Z_{i}^{(+)}\leftrightarrow Z_i^{(-)}$ for $i=1,\ldots,N-1$, and $D_N\to -E_N$ and $E_N\to -D_N$. 

Then, to self-test the Dicke states one proceeds in the following way:
\begin{enumerate}
\item Project any $(N-k-1)$-element subset $\mathcal{S}_i$ of the first $N-1$ parties of $\ket{\psi}$
(there are $\binom{N-1}{N-1-k}$ 
such subsets)
onto $\bigotimes_{j\in \mathcal{S}_i} Z_{j}^{(+)}$
and check whether the state corresponding to the remaining parties satisfies
the conditions for $|xD_{k+1}^k\rangle=\sigma^{\otimes (k+1)}_x|xW_{k+1}\rangle$.

\item For every sequence $(i_1,\ldots,i_N)$
consisting of $k+1$ ones on the first $N-1$ positions,  check that the state $\ket{\psi}$ obeys the following correlations 
\begin{equation}\label{VinaBujanda}
\langle\psi|Z_1^{(i_1)}\otimes\ldots\otimes Z_N^{(i_N)}|
\psi\rangle=0,
\end{equation}
where $Z_i^{(\tau_i)} = \left[\mathbbm{1} + (-1)^{\tau_i}Z_i\right]/2$.
\end{enumerate}
The detailed proof that the above procedure allows to self-test the Dicke states
is presented in \ref{App:Dicke}.

Notice that our self-test exploits two observables per site and the total number of correlators one has to determine for every Dicke state in this procedure again scales linearly with $N$, in contrast with the exponential scaling of quantum state tomography.

\subsection{Graph states} 
\label{subsec:graph}

We finally demonstrate that our method applies also to the graph states. These are $N$-qubit quantum states that have been widely exploited in quantum information processing, in particular in quantum computing, error correction, and secret sharing (see, e.g., Ref. \cite{GraphReview}). It is thus an interesting question to design efficient methods of their certification, in particular self-testing. Such a method was proposed in Ref. \cite{McKague} however, in general it needs three measurements for at least one party. Below we show that the approach based on violation of the CHSH Bell inequality provides a small improvement, as it requires only two measurements at each site.

Before stating our result, we introduce some notation. Consider a graph $G=(V,E)$ with $V$ and $E$ denoting respectively the $N$-element set of vertices of $G$ and the set of edges connecting elements of $V$. A graph state corresponding to $G$ is an $N$-qubit state given by $\ket{\psi_G}=\prod_{(a,b)\in E}U_{a,b}\ket{+}^{\otimes N}$, where $U_{a,b}$ is the controlled-$Z$ interaction between qubits $a$ and $b$, the product goes over all edges of $G$, and $\ket{+}=(\ket{0}+\ket{1})/\sqrt{2}$. Notice that $\ket{\psi_G}$ can also be written as 
\be\label{eq:GraphState}
\ket{\psi_G} = \frac{1}{\sqrt{2^N}}\sum_{\mathbf{i} \in \{0,1\}^n}(-1)^{\mu(\mathbf{i})}\ket{\mathbf{i}},
\ee 
where the sum is over all sequences $\mathbf{i}=(i_1,\ldots,i_N)$ with each $i_j\in\{0,1\}$, and 
$\mu(\mathbf{i})$ is the number of edges connecting qubits being in the state $\ket{1}$ for a given ket $\ket{\mathbf{i}}$. 

The main property of the graph states underlying our self-test is that by measuring all the neighbours of a pair of connected qubits $i,j$ in the $\sigma_z$-basis, the two qubits $i$ and $j$ are left in one of the Bell states (cf. Ref. \cite{HeinPRA}):
\begin{equation}\label{stanyBella}
\frac{1}{\sqrt{2}}(\sigma_z^{m_i}\otimes \sigma_{z}^{m_j})
(\ket{0+}+\ket{1-})
\end{equation}
where $m_i$ is the number of parties from $\nu_{i,j} \setminus \{j\}$ whose result of the measurement in the $\sigma_z$-basis was $-1$. In (\ref{stanyBella}) we neglect an unimportant $-1$ factor that might appear.

Having all this, we can now state formally our result. Given a graph $G$ and the corresponding graph state $\ket{\psi_G}$, let $\nu_i$ denote the set of of all neighbours of the qubit $i$ 
(all qubits connected to $i$ by an edge). Likewise, we denote by $\nu_{i,j}$ the set of neighbours of qubits $i$ and $j$. Let then $|\nu_i|$ and $|\nu_{i,j}|$ be the numbers of elements of $\nu_i$ and $\nu_{i,j}$, respectively. Also, for simplicity, we label 
the qubits of $\ket{\psi_G}$ in such a way that the qubits $N-1$ and $N$ are connected and
the qubit $N$ is the one with the smallest number of neighbours. Denoting $Z^{(\tau)}_{\nu_{i,j}}=\otimes_{l\in\nu_{i,j}}Z_l^{(\tau_l)}$,
where $\tau$ is an $|\nu_{i,j}|$-element sequence with each $\tau_l\in \{0,1\}$ (the operator $Z^{(\tau)}_{\nu_{i,j}}$ acts only on the parties belonging to $\nu_{i,j}$), we can state our result.

\begin{theorem}\label{theorem3}
Let $\ket{\psi}$ and measurements $Z_i,X_i$ with $i=1,\ldots,N-1$
and $D_N,E_N, Z_N \equiv (D_N-E_N)/\sqrt{2},_N \equiv (D_N+E_N)/\sqrt{2}$,  satisfy $\langle Z^{(\tau)}_{\nu_{N-1,N}}\rangle=1/2^{|\nu_{N-1,N}|}$, and
\begin{equation}\label{Graphcond1}
\left\langle Z^{(\tau)}_{\nu_{N-1,N}}\otimes B_{N-1,N}^{(m_{N-1},m_N)}\right\rangle=\frac{2\sqrt{2}}{2^{|\nu_{N-1,N}|}}
\end{equation} 
for every choice of the $|\nu_{i,j}|$-element sequence $\tau$. The Bell operators $B_{N-1,N}^{(m_{N-1},m_{N})}$ are defined as
%
$B_{N-1,N}^{(m_{N-1},m_{N})}=(-1)^{m_N} X_{N-1}\otimes (D_N+E_N)+(-1)^{m_{N-1}}Z_{N-1}\otimes (D_N-E_N).$
%
Additionally, we assume that 
\begin{eqnarray}\label{GraphCond2}
\left\langle Z^{(\tau)}_{\nu_{i,j}}\right\rangle=\frac{1}{2^{|\nu_{i,j}|}},\quad
\left\langle Z^{(\tau)}_{\nu_{i,j}}\otimes Z_i\otimes X_j\right\rangle=\frac{(-1)^{m_{j}}}{2^{|\nu_{i,j}|}}&&
\end{eqnarray}
for all connected pairs of indices $i\neq j$ except for $(N-1,N)$. Then $\Phi_N(\ket{\psi})=\ket{\psi_G}$. 
%
\end{theorem}
The proof of this statement may be found in \ref{App:Graph}. It is worth noting that the above approach exploits violations of the CHSH Bell inequality between a single pair of parties [cf. (\ref{Graphcond1})], and not between every pair of neighbours.

\section{Conclusion and discussion}\label{sec:conc}

We investigated a simple, but potentially general, approach to self-testing multipartite states, inspired by \cite{Singapur}, which relies on the well understood method of self-testing bipartite qubit states based on the maximal violation of the tilted CHSH Bell inequality. This approach allows one to self-test, with few measurements, all permutationally-invariant Dicke states, all partially entangled GHZ qubit states, and to recover self-testing of graph states (which was previously known through stabilizer-state methods). In our work, we also generalize self-testing of partially entangled GHZ qubit states to the qudit case, using techniques from \cite{Coladangelo}. We obtain the first self-testing result for a class of multipartite qudit states, by showing that all multipartite qudit states that admit a Schmidt decomposition can be self-tested. Importantly, our self-tests have a low complexity in terms of resources as they require up to four measurement choices per party, and the total number of correlators that one needs to determine scales linearly with the number of parties. {\ivan{This is because, in contrast with tomographic methods, one does not need to check correlators between all possible measurements of each party, but a smaller number of correlators already imposes a rigid structure on the state.}} \\ 

As a direction for future work, we are particularly interested in extending this approach to self-test any generic multipartite entangled state of qubits (which is local-unitary equivalent to its complex conjugate in any basis). The main challenge here is to provide a general recipe to construct a single isometry that self-tests the global state from the different ones derived from various subtests (i.e. from projecting various subsets of parties and looking at the correlations of the remaining ones). This appears to be challenging for states that do not have any particular symmetry.

In this paper we have not made any estimations regarding robustness of the presented self-testing protocols. The standard methods based on norm inequalities (see \cite{MYS,McKague,Cedric,SASA}) can be applied to obtain robustness bounds. However, calculating them is a tedious task which for states of many particles generally does not lead to practically useful bounds. Methods for finding better bounds, such as those presented in \cite{swap,jed}, are not easily applicable to multipatite states. In general, finding methods to estimate robustness of self-testing of multipartite states remains as one of the open questions.

Finally, notice that all presented self-tests which rely on the maximal violation of the CHSH Bell inequality can be restated and proved in terms of the other available self-tests. In particular, any self-test discussed in \cite{Wang} would work in case of two measurements per site, and self-tests in \cite{SASA} would work for higher number of inputs.\\

\textit{Note added:} After finishing this work, we learned about works \cite{WuPhD} and \cite{Fad}, where the authors obtain self-testing of $N$-partite W-states and Dicke states, respectively.  \\

\ack The authors thank Flavio Baccari, Marc Roda, Alexia Salavrakos and Thomas Vidick for useful discussions. This work was supported by Spanish MINECO (QIBEQI FIS2016-80773-P and Severo Ochoa SEV-2015-0522), the AXA Chair in Quantum Information Science, Generalitat de Catalunya (CERCA Programme), Fundaci\'{o} Privada Cellex and ERC CoG QITBOX. This project has received funding from the European Union's Horizon 2020 research and innovation programme under the Marie Sk\l{}odowska-Curie grant agreement No 705109. I. \v{S}. acknowledges the support of "Obra Social La Caixa 2016". A.C. is supported by AFOSR YIP award number FA9550-16-1-0495.
R. A. acknowledges the support from the Foundation for Polish Science 
through the First Team project No First TEAM/2017-4/31 co-financed by the European Union under
the European Regional Development Fund.

\appendix

%

%
%

\section{Proof of Theorem \ref{thm-PEGHZ}}
\label{appendix_C}

For ease of exposition, we prove the Theorem in the case $N=4$, with the extension to general $N$ being immediate. 

Let $A_0,A_1,B_0,B_1,C_0,C_1,D_0,D_1$, be the pairs of observables for the four parties. For an observable $D$, let $P_D^{a} = [\mathds{1} +(-1)^a D]/2$, and for brevity let $c_{\theta}$ and $s_{\theta}$ denote respectively $\cos{\theta}$ and $\sin{\theta}$. For clarity, we recall the correlations from Theorem \ref{thm-PEGHZ}, for the case $N = 4$:
%
\begin{eqnarray} \label{prva}
\fl \hspace{-1cm}\bra{\psi}P_{A_0}^0\ket{\psi} = \bra{\psi}P_{B_0}^0\ket{\psi} = \bra{\psi}P_{C_0}^0\ket{\psi} = \bra{\psi}P_{A_0}^0P_{C_0}^0\ket{\psi} =  \bra{\psi}P_{B_0}^0P_{C_0}^0\ket{\psi} = \Ct, \\ \label{druga}
\fl \hspace{-1cm}\bra{\psi}P_{A_1}^a P_{B_1}^b \ket{\psi} = \frac{1}{4}, \qquad \mbox{for $a,b \in {0,1}$}
\end{eqnarray}
\begin{equation}
\fl \hspace{-1cm} \bra{\psi}P_{A_1}^a P_{B_1}^b \big(\alpha C_0 + C_0D_0 + C_0D_1 + (-1)^{a+b}(C_1D_0 - C_1D_1)\big)\ket{\psi} = \frac{\sqrt{8+2\alpha^2}}{4},\quad \mbox{for $a,b \in {0,1}$} \label{treca}
\end{equation}
%

where $\tan{2\theta} = (\frac{2}{\alpha^2}-\frac{1}{2})^{\frac{1}{2}}$. Equations (\ref{prva}) imply, by Cauchy-Schwartz inequality, that
\begin{equation}
P_{A_0}^0\ket{\psi} = P_{B_0}^0\ket{\psi} = P_{C_0}^0\ket{\psi}
\end{equation}
and consequently
\begin{equation}
P_{A_0}^1\ket{\psi} = P_{B_0}^1\ket{\psi} = P_{C_0}^1\ket{\psi}.
\end{equation}

Notice that equation (\ref{druga}) implies $\|P_{A_1}^aP_{B_1}^b\ket{\psi}\| = 1/2$, for $a,b \in \{0,1\}$, and that the equations in (\ref{treca}) describe maximal violations of tilted CHSH inequalities by the normalized state $2P_{A_1}^aP_{B_1}^b\ket{\psi}$, for $a,b \in \{0,1\}$ (the ones for $a\oplus b =1 $ are tilted CHSH inequalities upon relabelling $D_1 \rightarrow -D_1$).

Let $\mu$ be such that $\tan{\mu} = s_{2\theta}$. Define $X_A := A_1, X_B := B_1$ and $X_C := C_1$. Then, let $Z_D' = (D_0+D_1)/2\cos{\mu}$, and let $Z_D^*$ be $Z_D'$ where we have replaced the zero eigenvalues with $1$. Define $Z_D = Z_D^*|Z_D^*|^{-1}$. Define $X_D$ similarly starting from $X_D' =(D_0-D_1)/2\cos{\mu}$. Let $P_{Z_D}^a = [\mathds{1}+ (-1)^a Z_D]/2$. The maximal violations of tilted CHSH from (\ref{treca}) imply, thanks to Lemma \ref{LemmaTilt}, that
\begin{eqnarray}\label{cetvrta}
P_{C_0}^a = P_{Z_D}^a, \,\,\,\,\,\,\,\, \mbox{for $a\in \{0,1\}$},\\
\st P_{A_1}^a P_{B_1}^bX_CX_DP_{C_0}^0\ps  = (-1)^{a+b}\ct P_{A_1}^aP_{B_1}^b P_{C_0}^1\ps, \,\,\, \mbox{for $a,b\in \{0,1\}$}. \label{c5}
\end{eqnarray}
If we introduce notation $X_A = A_1, X_B = B_1$ and $X_C = C_1$, then
\ba \nonumber
\fl X_AX_BX_CX_D P_{A_0}^{1}\ps &=& (P_{A_1}^0-P_{A_1}^1)(P_{B_1}^0-P_{B_1}^1)X_CX_DP_{C_0}^1\ps\\ \nonumber
&=& P_{A_1}^0P_{B_1}^0X_CX_DP_{C_0}^1\ps - P_{A_1}^0P_{B_1}^1X_CX_DP_{C_0}^1\ps - P_{A_1}^1P_{B_1}^0X_CX_DP_{C_0}^1\ps\nonumber\\
&& + P_{A_1}^1P_{B_1}^1X_CX_DP_{C_0}^1\ps\\ \nonumber
&=& \frac{\st}{\ct} P_{A_1}^0P_{B_1}^0P_{A_0}^0\ps + \frac{\st}{\ct} P_{A_1}^0P_{B_1}^1P_{A_0}^0\ps + \frac{\st}{\ct} P_{A_1}^0P_{B_1}^1P_{A_0}^0\ps + \frac{\st}{\ct} P_{A_1}^1P_{B_1}^1P_{A_0}^0\ps  \\ 
&=& \frac{\st}{\ct}P_{A_0}^0\ps,
\ea
where we used equation (\ref{c5}) to obtain the third line, and $\sum_{a,b \in \{0,1\}} P_{A_1}^aP_{B_1}^b = \mathds{1}$ to obtain the last. Conditions (\ref{eq12}) and (\ref{eq13}) of Theorem \ref{thm-PEGHZ} follow immediately from the above. Note that we omitted the proof that $Z_D$ and $X_D$ act on $\ket{\psi}$ in the same way as $Z'_D$ and $X'_D$, respectively. The proof relies on the standard self-testing procedure named regularization. For a detailed proof see \cite{Cedric} or \cite{SASA}.

\section{Proof of Lemma \ref{yngen}}
In this section, we provide a proof of Lemma \ref{yngen}. We explicitly construct a local isometry $\Phi$ such that $\Phi(\ket{\psi})=\ket{aux}\otimes\ket{\Psi}$ for any Schmidt state $\ket{\Psi} = \sum_{j=0}^{d-1} c_j \ket{j}^{\otimes N}$, where $0< c_j < 1$ for all $j$ and $\sum_{j=0}^{d-1}c_j^2 = 1$, and $\ket{aux}$ is some auxiliary state. 

\begin{proof}
Recall that $\{P_l^{(k)}\}_{k=0}^{d-1}$ are complete sets of orthogonal projections for $l = 1,\ldots,N-1$ by hypothesis. Then, notice that for $i \neq j$ we have, using condition (\ref{c2}), $P_N^{(i)}P_N^{(j)} \ps = P_N^{(i)}P_1^{(j)} \ps =P_1^{(j)}P_1^{(i)} \ps = 0$, i.e., the $P_N^{(k)}$ are ``orthogonal when acting on $\ps$''. 

Let $\mathcal{A}$ be the unital algebra generated by $\{P_1^{(k)}\}$. Let $\mathcal{H}' = \mathcal{A} \ps$, where $\mathcal{A} \ps = \{Q\ps : Q \in \mathcal{A}\}$. Let $\tilde{P}_N^{(k)} = P_N^{(k)}|_{\mathcal{H}'}$ be the restriction of $P_N^{(k)}$ to $\mathcal{H}'$. Then, $\{\tilde{P}_N^{(k)}\}_{k=0}^{d-1}$ is a set of orthogonal projections. This is because, thanks to (\ref{c2}), one can always move the relevant operators to be in front of $\ps$, as in the simple example 
\begin{equation}
\tilde{P}_N^{(i)}\tilde{P}_N^{(j)} (P_1^{(k)}\ps') = P_1^{(k)}\tilde{P}_N^{(i)}\tilde{P}_N^{(j)}  \ps = 0.
\end{equation}

Thus, the set $\{\tilde{P}_B^{(k)}, I-P'_B\}$, where $P'_B$ is the sum of all other projections, is a complete set of orthogonal projections.

Now, define $Z_{l} := \sum_{k=0}^{d-1} \omega^k P_l^{(k)}$, for $l=1,\ldots,N-1$, and $Z_{N} := \sum_{k=0}^{d-1} \omega^k \tilde{P}_{N}^{(k)} + \mathds{1} - \sum_{k=0}^{d-1} \tilde{P}_{N}^{(k)}$.
In particular, the $Z_l$ are all unitary. Notice, moreover, that $\big(\mathds{1} - \sum_k \tilde{P}_N^{(k)}\big) \ps = 0$, by using (\ref{c2}) and the fact that the $\{P_l^{(k)}\}$ are complete. 

Define the local isometry
\begin{equation}
\Phi := \bigotimes_{l=1}^{N} R_{ll'} \bar{F}_{l'}S_{ll'}F_{l'} \mbox{App}_l,
\end{equation}
where $\mbox{App}_l: \mathcal{H}_l \rightarrow \mathcal{H}_l \otimes \mathcal{H}_{l'}$ is the isometry that simply appends $\ket{0}_l'$, $F$ is the quantum Fourier transform, $\bar{F}$ is the inverse quantum Fourier transform, $R_{ll'}$ is defined so that $\ket{\phi}_{l} \ket{k}_{l'} \mapsto X^{(k)}_{l}\ket{\phi}_{l} \ket{k}_{l'}\,\,\, \forall \ket{\phi}$, and $S_{ll'}$ is defined so that $\ket{\phi}_{l} \ket{k}_{l'} \mapsto Z^{k}_{l}\ket{\phi}_{l}\ket{k}_{l'} \,\,\, \forall \ket{\phi}$. We compute the action of $\Phi$ on $\ket{\psi}$. For ease of notation with drop the tildes from the $\tilde{P}_N^{(k)}$, while still referring to the new orthogonal projections. 
\begin{eqnarray} \nonumber
\fl \ket{\psi}\otimes \ket{0}^{\otimes N} \stackrel{\bigotimes_l F_{l'}}{\longrightarrow} \frac{1}{d^{N/2}}\sum_{k_1,\ldots,k_N}\ket{\psi}\otimes \bigotimes_l \ket{k_l}_{l'} \\
\stackrel{\bigotimes_l S_{ll'}}{\longrightarrow} \frac{1}{d^{N/2}}\sum_{k_1,\ldots,k_N} \left[\prod_{i=1}^{N-1}\left(\sum_{j_i}\omega^{j_i}P^{(j_i)}_{i}\right)^{k_i}\right]\left(\sum_{j_N}\omega^{j_N}P^{(j_N)}_{N} + \mathds{1} - \sum_k P_N^{(j_N)}\right)^{k_N}\ket{\psi}\otimes \bigotimes_l \ket{k_l}_{l'} \nonumber\\
= \frac{1}{d^{N/2}}\sum_{k_1,\ldots,k_N}\sum_{j_1,\ldots,j_N}\prod_{i=1}^{N}\omega^{j_ik_i}P^{(j_i)}_{i}\ket{\psi}\otimes \bigotimes_l \ket{k_l}_{l'} \nonumber\\
= \frac{1}{d^{N/2}}\sum_{k_1,\ldots,k_N}\sum_{j_1,\ldots,j_N}\prod_{i=1}^{N}\omega^{j_ik_i}P^{(j_i)}_{1}\ket{\psi}\otimes \bigotimes_l \ket{k_l}_{l'} \nonumber\\
= \frac{1}{d^{N/2}}\sum_{k_1,\ldots,k_N}\sum_{j}\omega^{j(\sum_i k_i)}P^{(j)}_{1}\ket{\psi}\otimes \bigotimes_l \ket{k_l}_{l'} \nonumber\\
\stackrel{\bigotimes_l \bar{F}_{l'}}{\longrightarrow}\frac{1}{d^N}\sum_{k_1,\ldots,k_N}\sum_{j}\sum_{m_1,\ldots,m_N}\omega^{j(\sum_i k_i)}\prod_r\omega^{-m_rk_r}P^{(j)}_{1}\ket{\psi}\otimes \bigotimes_l \ket{m_l}_{l'}\nonumber \\
=\frac{1}{d^N}\sum_{k_1,\ldots,k_N}\sum_{j}\sum_{m_1,\ldots,m_N}\prod_i \omega^{k_i(j-m_i)}P^{(j)}_{1}\ket{\psi}\otimes \bigotimes_l \ket{m_l}_{l'}\nonumber \\
= \sum_{j}P^{(j)}_{1}\ket{\psi}\otimes \ket{j}^{\otimes N}\label{A12}\\
\stackrel{\bigotimes_l R_{ll'}}{\longrightarrow} \sum_{j}\left(\prod_iX^{(j)}_{i}\right)P^{(j)}_{1}\ket{\psi}\otimes \ket{j}^{\otimes N} \nonumber \\
= \sum_{j}\frac{c_j}{c_0} P^{(0)}_{1}\ket{\psi}\otimes \ket{j}^{\otimes N} \label{E12}\\
= \frac{1}{c_0}P^{(0)}_{1}\ket{\psi} \otimes \sum_{j} c_j \ket{j}^{\otimes N} \nonumber\\
= \ket{\textrm{aux}} \otimes \ket{\Psi}, \nonumber
\end{eqnarray}
\noindent where to get (\ref{E12}) we used condition (\ref{c2}). 
It is an easy check to see that the whole proof above can be repeated by starting from a mixed joint state, yielding a corresponding version of the Lemma that holds for a general mixed state.
\end{proof}

\section{Proof of Theorem \ref{thm_schmidt}}
\label{appendix_D}
As mentioned, we work in the tripartite case, as the general $n$-partite case follows analogously. The measurements of Alice, Bob and Charlie can be assumed to be projective, since we make no assumption on the dimension of the system. For ease of notation, the proof assumes that the joint state is pure, but one easily realizes that the proof goes through in the same way by rephrasing everything in terms of density matrices (see \cite{Coladangelo} for a slightly more detailed discussion).

Let $\ps$ be the unknown joint state, and let $P_{A_x}^a$ be the projection on Alice side corresponding obtaining outcome $a$ on question $x$. Define $P_{B_y}^b$ and $P_{C_z}^c$ similarly on Bob and Charlie's side. The proof structure follows closely that of \cite{Coladangelo}, and goes through explicitly constructing projectors and unitary operators satisfying the sufficient conditions of Lemma \ref{yngen}. 

Define $\Ah_{x,m} = P^{2m}_{A_x}-P^{2m+1}_{A_x}$, $\Bh_{y,m} = P^{2m}_{B_y}-P^{2m+1}_{B_y}$ and $\Ch_{z,m} = P^{2m}_{C_z}-P^{2m+1}_{C_z}$, for $x,y,z \in \{0,1\}$. Let $\mathds{1}_{A_x}^m = P^{2m}_{A_x}+P^{2m+1}_{A_x}$ 
and similarly define $\mathds{1}_{B_y}^m$ and $\mathds{1}_{C_z}^m$ for $x,y,z \in \{0,1\}$. Now, 
\begin{eqnarray}
\|P_{A_0}^{2m}\| &=& {\bracket{\psi}{P_{A_0}^{2m}}{\psi}}^{\frac{1}{2}}\nonumber\\ 
&=& \left(\bracket{\psi}{P_{A_0}^{2m} \sum_{i=0}^{d-1}P_{B_0}^{i} \sum_{j=0}^{d-1}P_{C_0}^{j}}{\psi}\right)^{\frac{1}{2}} \nonumber\\
&=& c_{2m},
\end{eqnarray}
and $\|P_{A_0}^{2m+1}\| = c_{2m+1}$. Similarly, we derive $\| \mathds{1}_{A_x}^m \ps\| = \| \mathds{1}_{B_y}^m \ps \| = \| \mathds{1}_{C_z}^m \ps \| = (c_{2m}^2+c_{2m+1}^2)^{1/2}$ for any $m$ and $x,y,z \in \{0,1\}$. Notice then that 
\begin{eqnarray}
\bracket{\psi}{\mathds{1}_{A_x}^m\mathds{1}_{B_y}^m}{\psi} &=& \bracket{\psi}{\mathds{1}_{A_x}^m\mathds{1}_{B_y}^m \sum_{i=0}^{d-1}P_{C_0}^{i}}{\psi}\nonumber\\ 
&=& \bracket{\psi}{\mathds{1}_{A_x}^m\mathds{1}_{B_y}^m \mathds{1}_{C_0}^m}{\psi}\nonumber\\
 &=& c_{2m}^2+c_{2m+1}^2,
\end{eqnarray}
where the second last equality is from the block-diagonal structure of the correlations. Since $\| \mathds{1}_{A_x}^m \ps\| = \| \mathds{1}_{B_y}^m \ps \| = (c_{2m}^2+c_{2m+1}^2)^{1/2}$, then Cauchy-Schwartz inequality implies $\mathds{1}_{A_x}^m \ps = \mathds{1}_{B_y}^m \ps$. So, we have
\begin{equation}
\label{E1}
\mathds{1}_{A_x}^m \ps = \mathds{1}_{B_y}^m \ps = \mathds{1}_{C_z}^m \ps 
\end{equation}
for all $x,y,z \in \{0,1\}$. The correlations are, by design, such that $\Ah_{0,m}, \Ah_{1,m}, \Bh_{0,m}, \Bh_{1,m}, \Ch_{0,m}, \Ch_{0,m}$, the associated projections $P_{A_i}^{j},P_{B_i}^{j},P_{C_i}^{j}$, $j \in \{2m,2m+1\}$ and $\ps$ reproduce the correlations $(c_{2m}^2+c_{2m+1}^2) \cdot C_{x,y,z}^{\mbox{ghz}_{3,2,\theta_m}}$. In order to apply Theorem \ref{thm-PEGHZ}, we need to define the normalized state $\ket{\psi'_m} := (\mathds{1}_{A_0}^m \ps)/(c_{2m}^2+c_{2m+1}^2)^{1/2}$ and the ``unitarized'' 
versions of the operators above, namely $\hat{D}^{\mathscr{u}}_{i,m} := \mathds{1}-\mathds{1}_m^{D_i} + \hat{D}_{i,m}$, for $D \in \{A,B,C\}$. It is easy to check that then $\hat{A}^{\mathscr{u}}_{i,m}, \hat{B}^{\mathscr{u}}_{i,m}$ and $\hat{C}^{\mathscr{u}}_{i,m}$ satisfy the conditions of Theorem \ref{thm-PEGHZ} (for $N = 3$) on state $\ket{\psi'_m}$. Thus, we have 
\begin{eqnarray}
Z^{\mathscr{u}}_{A,m} \ket{\psi'_m} = Z^{\mathscr{u}}_{B,m} \ket{\psi'_m} = Z^{\mathscr{u}}_{C,m} \ket{\psi'_m}, \label{E2} \\
X^{\mathscr{u}}_{A,m} X^{\mathscr{u}}_{B,m} X^{\mathscr{u}}_{C,m} (\mathds{1} - Z^{\mathscr{u}}_{A,m}) \ket{\psi'_m} = \tan(\theta_m) (\mathds{1} + Z^{\mathscr{u}}_{A,m}) \ket{\psi'_m}. \label{E3}
\end{eqnarray}
Define the subspace $\mathcal{C}_m=\mbox{range}(\mathds{1}_m^{C_0}) + \mbox{range}(\mathds{1}_m^{C_1})$, and the projection $\mathds{1}_{\mathcal{C}_m}$ onto subspace $\mathcal{C}_m$. Then, notice from the way $Z^{\mathscr{u}}_{C,m}$ is defined, that it can be written as $Z^{\mathscr{u}}_{C,m} = \mathds{1}-\mathds{1}_{\mathcal{C}_m} + \tilde{Z}_{C,m}$, where $\tilde{Z}_{C,m}$ is some operator living entirely on subspace $\mathcal{C}_m$. This implies that $Z^{\mathscr{u}}_{C,m} \ket{\psi_m} =  \tilde{Z}_{C,m} \ket{\psi_m} = \tilde{Z}_{C,m} \ps$, where we have used (\ref{E1}) and the fact that 
\begin{eqnarray}\label{onesb}
\mathds{1}_m^{C_0} \ket{\psi} = \mathds{1}_m^{C_1} \ket{\psi}&\Longrightarrow& \mathds{1}_{\mathcal{C}_m} \ket{\psi} = \mathds{1}_m^{C_i}\ket{\psi}\,.
\end{eqnarray}
Hence, from (\ref{E2}) it is not difficult to deduce that $\hat{A}_{0,m} \ps = \hat{B}_{0,m} \ps = \tilde{Z}_{C,m} \ps$. \\

\noindent \textbf{Constructing the projections of Lemma \ref{yngen}.} Define projections $P_A^{(2m)} := (\mathds{1}_m^{A_0} + \hat{A}_{0,m})/2 = P_{A_0}^{2m}$, $P_A^{(2m+1)} := (\mathds{1}_m^{A_0} - \hat{A}_{0,m})/2 = P_{A_0}^{2m+1}$, $P_B^{(2m)} :=(\mathds{1}_m^{B_0} + \hat{B}_{0,m})/2 = P_{B_0}^{2m}$, $P_B^{(2m+1)} := (\mathds{1}_m^{B_0} - \hat{B}_{0,m})/2 = P_{B_0}^{2m+1}$,
$P_C^{(2m)} := (\mathds{1}_{\mathcal{C}_m} + \tilde{Z}_{C,m})/2$ and $P_C^{(2m+1)} := (\mathds{1}_{\mathcal{C}_m} - \tilde{Z}_{C,m})/2$. 

Note that $P_C^{(2m)},P_C^{(2m+1)}$ are indeed projections, since $\tilde{Z}_{C,m}$ has all $\pm 1$ eigenvalues corresponding to subspace $\mathcal{C}_m$, and is zero outside. We also have, for all $m$ and $k=2m,2m+1$,
\begin{eqnarray}
\label{E5}
\fl P_B^{(k)} \ps = P_A^{(k)} \ps = \frac{1}{2}[\mathds{1}_m^{A_0} +(-1)^k \hat{A}_{0,m}] \ps &=&  \frac{1}{2}[\mathds{1}_m^{B_0} +(-1)^k \hat{A}_{0,m}] \ps \\ \nonumber
&=& \frac{1}{2}[\mathds{1}_{\mathcal{B}_m} +(-1)^k \tilde{Z}_{B,m}] \ps = P_C^{(k)} \ps.
\end{eqnarray}

Further, notice that $[\mathds{1} +(-1)^k Z^{\mathscr{u}}_{A,m}] \ket{\psi'_m} =  [\mathds{1}_m^{A_0} +(-1)^k\hat{A}_{0,m}] \ket{\psi'_m} = [\mathds{1}_m^{A_0} +(-1)^k \hat{A}_{0,m}] \ps = P_A^{(k)}\ps$. Substituting this into (\ref{E3}), gives
\begin{equation}
\label{E6}
X^{\mathscr{u}}_{A,m}X^{\mathscr{u}}_{B,m} X^{\mathscr{u}}_{C,m}  P_A^{(2m+1)} \ps = \tan(\theta_m) P_A^{(2m)} \ps  = \frac{c_{2m+1}}{c_{2m}}P_A^{(2m)} \ps. \\
\end{equation}

Now, for the "shifted" blocks, we can similarly define $\Ah_{x,m}'$, $\Bh_{x,m}'$ and $\Ch_{x,m}'$ as $\Ah_{x,m} = P^{2m+1}_{A_x}-P^{2m+2}_{A_x}$ and similar. Then, analogously, we deduce the existence of hermitian and unitary operators $Y'_{A,m}$, $Y'_{B,m}$ and $Y'_{C,m}$ such that 
\begin{equation}
Y_{A,m} Y_{B,m} Y_{C,m} P_{A}^{(2m+2)} \ps =\frac{c_{2m+2}}{c_{2m+1}} P_{A}^{(2m+1)} \ps. \label{E7}
\end{equation}

\noindent \textbf{Constructing the unitary operators of Lemma \ref{yngen}.} We will now directly construct unitary operators satisfying conditions (\ref{c2},\ref{c3}) of Lemma \ref{yngen}. Define $X_{A/B/C}^{(k)}$ as follows:
\begin{equation}
X_{A}^{(k)} =  
\cases{
\mathds{1},    & \text{if } k=0,\\
X_{A,0}Y_{A,0}X_{A,1}Y_{A,1} \ldots X_{A,m-1}Y_{A,m-1}X_{A,m},   & \text{if } k=2m+1,\\
X_{A,0}Y_{A,0}X_{A,1}Y_{A,1} \ldots X_{A,m-1}Y_{A,m-1},              & \text{if } k=2m,}
\end{equation}
and analogously for $X_{B}^{(k)}$ and $X_{C}^{(k)}$.
Note that $X_{A}^{(k)}$ and $X_{B}^{(k)}$ are unitary since they are product of unitaries.
Finally, we are left to check that 
\begin{equation}
X^{(k)}_{A}X^{(k)}_{B}X^{(k)}_{C} P^{(k)}_{A}\ket{\psi} =\frac{c_k}{c_0} P^{(0)}_{A}\ket{\psi}. \label{eqfinal}
\end{equation}
The case $k=0$ holds trivially. For $k=2m+1$,
For $k=2m+1$, 
\begin{eqnarray}
\fl X^{(k)}_{A}X^{(k)}_{B}X^{(k)}_{C} P^{(k)}_{A}\ket{\psi} \nonumber \\
=X_{A,0}Y_{A,0}X_{B,0}Y_{B,0}X_{C,0}Y_{C,0}\ldots X_{A,m-1}Y_{A,m-1}X_{B,m-1}Y_{B,m-1}X_{C,m-1}Y_{C,m-1}\nonumber\\
\times X_{A,m}X_{B,m}X_{C,m} P_{A}^{(2m+1)} \ps \nonumber\\
\stackrel{\ref{E6}}{=}\frac{c_{2m+1}}{c_{2m}}X_{A,0}Y_{A,0}X_{B,0}Y_{B,0}X_{C,0}Y_{C,0}\ldots X_{A,m-1}Y_{A,m-1}X_{B,m-1}Y_{B,m-1}X_{C,m-1}Y_{C,m-1} P_{A}^{(2m)} \ps
\nonumber\\
\stackrel{\ref{E7}}{=} \frac{c_{2m+1}}{c_{2m}}\cdot \frac{c_{2m}}{c_{2m-1}}X_{A,0}Y_{A,0}X_{B,0}Y_{B,0}X_{C,0}Y_{C,0}\ldots X_{A,m-2}Y_{A,m-2}X_{B,m-2}Y_{B,m-2}\nonumber\\
\times X_{C,m-2}Y_{C,m-2} P_{A}^{(2m-1)} \ps
\nonumber\\
= \ldots  \nonumber\\
= \frac{c_{2m+1}}{\cancel{c_{2m}}} \cdot \frac{\cancel{c_{2m}}}{\cancel{c_{2m-1}}} \ldots \frac{\cancel{c_2}}{\cancel{c_1}} \cdot \frac{\cancel{c_1}}{c_0} P_{A}^{(0)} \ps \nonumber\\
= \frac{c_{2m+1}}{c_0} P_A^{(0)}\ps 
\end{eqnarray}
which is indeed (\ref{eqfinal}) as $2m+1=k$. The case $k=2m$ is similar. This concludes the proof of Theorem \ref{thm_schmidt}. 

\section{Self-testing of the W states}
\label{App:W}

In this section we provide a detailed proof of self-testing 
of the $\ket{W_N}$ state
\begin{equation}\label{Wstate}
\ket{W_N}=\frac{1}{\sqrt{N}}(\ket{0\ldots 01}+\ket{0\ldots 0 10}+\ldots+\ket{10\ldots0}).
\end{equation}
For our convenience we show how to self-test the following 
unitarily equivalent state
\begin{equation}
\ket{xW_N}=\frac{1}{\sqrt{N}}(\ket{0\ldots 0}+\ket{0\ldots 0 11}+\ldots+\ket{10\ldots01}).
\end{equation}
which is obtained from $\ket{W_N}$ by applying $\sigma_x$ to the last qubit of 
$\ket{W_N}$. This is because $\ket{xW_N}$ can be written as
\begin{equation}\label{xW2}
\ket{xW_N}=\frac{1}{\sqrt{N}}\left[\ket{0}^{\otimes N-2}(\ket{00}+\ket{11})_{i,N}+\ket{\mathrm{rest}_i}\right],
\end{equation}
where $(\ket{00}+\ket{11})_{i,N}$ stands for the two-qubit maximally entangled state distributed between the parties $i$ and $N$ with $i=1,\ldots,N-1$, and 
the vectors $\ket{\mathrm{rest}_i}$ contain the remaining kets. This decomposition explains the conditions we impose below. 

Let us now prove the following theorem. 

\begin{theorem}\label{ThmW}
Assume that for a given state $\ket{\psi}$ and measurements
$Z_i,X_i$ for parties $i=1,\ldots,N-1$ and $D_N$ and $E_N$ for the last party, the following conditions are satisfied
\begin{equation}\label{ThmWcond1}
\left\langle \bigotimes_{l=1, l\neq i}^{N-1}Z_l^{(+)}\right\rangle=\frac{2}{N},\qquad
\left\langle \bigotimes_{l=1,l\neq i}^{N-1}Z_l^{(+)}\otimes B_{i,N}^{(+)}\right\rangle=\frac{4\sqrt{2}}{N},
\end{equation}
with $i=1,\ldots,N-1$, where, as before, $B_{i,N}^{(+)}$ is the Bell operator between the parties 
$i$ and $N$ corresponding to the CHSH Bell inequality
\begin{equation}
B_{i,N}^{(+)}=Z_i \otimes D_N+Z_i\otimes E_N+X_i\otimes D_N-X_i\otimes E_N.
\end{equation} 
Moreover, we assume that 
\begin{equation}\label{ThmWcond2}
\langle Z_i^{(-)}\rangle=\frac{1}{N},\qquad
\left\langle \bigotimes_{l=1,l\neq i}^{N-1}Z_i^{(+)}\otimes Z_i^{(-)}\right\rangle=\frac{1}{N}
\end{equation}
with $i=1,\ldots,N-1$.
Then, for the isometry $\Phi_N$ one has
\begin{equation}\label{Wfinal}
\Phi_N\left(\ket{\psi}\ket{0}^{\otimes N}\right)=\ket{\mathrm{aux}}\ket{xW_N}.
\end{equation}
\end{theorem}
\begin{proof}Denoting $Z_N=(D_N+E_N)/\sqrt{2}$ and $X_N=(D_N-E_N)/\sqrt{2}$, the action of the isometry can be explicitly written as
\begin{equation}\label{IsoW}
\Phi_N\left(\ket{\psi}\ket{0}^{\otimes N}\right)=\sum_{\tau \in\{0,1\}^N}X_1^{\tau_1}\ldots X_N^{\tau_N}Z_{1}^{(\tau_1)}\ldots Z_N^{(\tau_N)}\ket{\psi}\ket{\tau_1\ldots \tau_N},
\end{equation}
where $\tau=(\tau_1,\ldots,\tau_N)$ with each $\tau_i\in\{0,1\}$ and 
$Z_i^{(\tau_i)}=[\mathbbm{1}+(-1)^{\tau_i}Z_i]/2$.
%
%

It should be noticed that in general the operators $Z_N$ and $D_N$ might not be unitary, and one should consider $\widetilde{X}_N$ and $\widetilde{Z}_N$, which by constructions are unitary. However, as already explained in \ref{appendix_C}, their action on $\ket{\psi}$ is the same as the action of $X_N$ and $Z_N$, thus, for simplicity, we use these operators.

The first bunch of conditions (\ref{ThmWcond1}) implies that the norm of
\begin{equation}
\ket{\psi_i}=Z_1^{(+)}\ldots Z_{i-1}^{(+)}Z_{i+1}^{(+)}\ldots Z_{N-1}^{(+)}\ket{\psi}
\end{equation}
is $\sqrt{2/N}$, which together with the second set of conditions in  (\ref{ThmWcond1}) implies that the 
normalized states $\ket{\widetilde{\psi}_{i}}=\sqrt{N/2}\,\ket{\psi_i}$ violate maximally the CHSH Bell inequality between the parties $i$ and $N$ for $i=1,\ldots,N-1$. This, by virtue of what was said in Lemma \ref{LemmaTilt} , yields the following identities
\begin{eqnarray}
&&\label{W11}(Z_i-Z_N)\ket{\widetilde{\psi}_{i}}=0\\
&&\label{W12}[X_i(I+Z_N)-X_N(I-Z_i)]\ket{\widetilde{\psi}_{i}}=0\\
&&\label{W13}\{Z_i,X_i\}\ket{\widetilde{\psi}_{i}}=0.
\end{eqnarray}
They immediately imply that all terms in . (\ref{IsoW}) for which one element of $\tau$ equals one 
and the rest equal zero vanish. To see it explicitly, let $\tau_i=1$ and $\tau_{j}=0$ for $j\neq i$. Then, 
for this $\tau$, $\ket{\psi_{\tau}}=X_iZ_{i}^{(-)}Z_{N}^{(+)}\ket{\psi_i}$. Applying (\ref{W11}) to the latter and exploiting the fact that $Z_{i}^{(-)}Z_{i}^{(+)}=0$, one finally finds that $\ket{\psi_{\tau}}=0$.

Let us now consider those components of  (\ref{IsoW}) for which $\tau$ obeys $\tau_i=\tau_N=1$ with $i=1,\ldots,N-1$ and $\tau_{j}=0$ with $j\neq i,N$.
Then, the following chain of equalities holds 
\begin{eqnarray}
Z_1^{(+)}\ldots Z_{i-1}^{(+)}X_i Z_i^{(-)}Z_{i+1}^{(+)}\ldots Z_{N-1}^{(+)}X_NZ_{N}^{(-)}\ket{\psi}&=&X_iZ_{i}^{(-)} X_NZ_{N}^{(-)}\ket{\psi_i}\nonumber\\&=&X_iZ_{i}^{(-)} X_NZ_{i}^{(-)}\ket{\psi_i}\nonumber\\
&=&X_i Z_{i}^{(-)}X_iZ_{N}^{(+)}\ket{\psi_i}\nonumber\\
&=& Z_i^{(+)} Z_{N}^{(+)}\ket{\psi_i}\nonumber\\
&=&Z_1^{(+)}\ldots Z_{N}^{(+)}\ket{\psi},
\end{eqnarray}
where the second equality stems from (\ref{W11}), the third from  (\ref{W12}), and, finally, the fourth equality is a consequence of
the anticommutation relation (\ref{W13}) and the fact that $X_i^2=\mathbbm{1}$. 

With all this in mind it is possible to group the terms in 
(\ref{IsoW}) in the following way
\begin{equation}
\Phi_N(\ket{\psi}\ket{0}^{\otimes N})=\ket{\mathrm{aux}}\ket{xW_N}+\ket{\Omega},
\end{equation}
where $\ket{\mathrm{aux}}=\sqrt{N}\,Z_1^{(+)}\ldots Z_{N}^{(+)}\ket{\psi}$ and $\ket{\Omega}$ contains all those terms for which
$\tau$ contains more than two ones or exactly two ones
but $\tau_N=0$. 

Now, our aim is to prove that $\ket{\Omega}=0$. To this end we first notice that (\ref{ThmWcond2}) imply the following correlations
\begin{eqnarray}\label{Wcorrs}
\langle\psi|Z^{(-)}_{i}Z_{j}^{(+)}|\psi\rangle=\frac{1}{N},
\end{eqnarray}
where $i\neq j$ and $i,j=1,\ldots,N-1$.
This is a direct consequence of the fact that $Z_{i}^{(\pm)}\leq \mathbbm{1}$, which in turn implies that each of correlators in (\ref{Wcorrs}) 
is bounded from above 
by $\langle\psi|Z_i^{(-)}|\psi\rangle$ and from below by 
$\langle\psi| Z^{(+)}_{1}\ldots Z_{i-1}^{(+)} Z_i^{(-)} Z_{i+1}^{(+)}\ldots Z_{N-2}^{(+)}Z_{N-1}^{(+)}|\psi\rangle $ and both are assumed to equal $1/N$ [cf. (\ref{ThmWcond2})].

The first relation in  (\ref{Wcorrs}) together with 
(\ref{ThmWcond2}) and the fact that $Z_{j}^{(+)}+Z_{j}^{(-)}=\mathbbm{1}$ yields $\langle\psi| Z_{i}^{(-)}Z_{j}^{(-)}|\psi\rangle=0$. This, due to the fact that 
$Z_{i}^{(-)}Z_{j}^{(-)}$ is a projector, allows one to write 
\begin{equation}
Z_{i}^{(-)}Z_{j}^{(-)}\ket{\psi}=0
\end{equation}
for $i,j=1,\ldots,N-1$. This is enough to conclude that $\ket{\Omega}=0$, which when plugged into (\ref{xW2}), leads directly to (\ref{Wfinal}) because each component in $\ket{\Omega}$ has either three $\tau_i$ which equal 1, or two $\tau_i$ that equal one but then $\tau_N=0$. 

Since the self-test relies on the maximal violation of the CHSH Bell inequality
by a set of states $\ket{\widetilde{\psi}_i}$ $(i=1,\ldots,N-1)$, it also inherits self-testing of the optimal CHSH measurements, meaning that
\begin{eqnarray}
\Phi_N\left(Z_i\ket{\psi}\ket{0}^{\otimes N}\right)=\ket{\mathrm{aux}}\otimes \sigma_z^{(i)}\ket{xW_N}\nonumber\\
\Phi_N\left(X_i\ket{\psi}\ket{0}^{\otimes N}\right)=\ket{\mathrm{aux}}\otimes \sigma_x^{(i)}\ket{xW_N}
\end{eqnarray}
for $i=1,\ldots,N-1$, and
\begin{eqnarray}
\Phi_N\left(D_N\ket{\psi}\ket{0}^{\otimes N}\right)=\ket{\mathrm{aux}}\otimes \frac{\sigma_z^{(i)}+\sigma_x^{(i)}}{\sqrt{2}}\ket{xW_N},\nonumber\\
\Phi_N\left(E_N\ket{\psi}\ket{0}^{\otimes N}\right)=\ket{\mathrm{aux}}\otimes \frac{\sigma_z^{(i)}-\sigma_x^{(i)}}{\sqrt{2}}\ket{xW_N}.
\end{eqnarray}
This completes the proof.
\end{proof}

It should be noticed that our self-test of the $W$ state exploits two observables per site and, as in the case of the partially entangled GHZ state, the number of correlators one needs to determine is $2N$, and thus scales linearly with $N$.

\section{Complete self-test of all the symmetric Dicke state}
\label{App:Dicke}

Here we show that the above self-test of the $N$-partite W state can be used to construct a self-test of all the Dicke states. Let us recall that Dicke states can be written in the following way 
\begin{equation}
\ket{D_N^{m}}={\binom{N}{m}}^{-1/2}
\sum_{i}P_i(\ket{0}^{\otimes N-m}\ket{1}^{\otimes m}),
\end{equation}
where the sum is over all permutation of an $N$-element set and $m=0,\ldots,N$ (there is $N+1$ such states). Let us notice that a Dicke state with $m\leq \lfloor N/2\rfloor$ is unitarily equivalent to a Dicke state with $m\geq \lceil N/2\rceil$, i.e., $\ket{D_{N}^m}=\sigma_z^{\otimes N}\ket{D_{N}^{N-m}}$ for $m=0,\ldots,\lfloor N/2\rfloor$. For this reason below we consider the Dicke states with $m\geq \lfloor N/2\rfloor$. 

Interestingly, to self-test Dicke states we can exploit the previously demostrated self-test of the W state. This follows from the fact that any Dicke state can be written as
\begin{equation}
\ket{D_{N}^m}=
%
%
\frac{1}{\sqrt{N}}\left(\sqrt{N-m}\,\ket{0}\ket{D_{N-1}^m}+\sqrt{m}\,\ket{1}\ket{D_{N-1}^{m-1}}\right)
\end{equation} 
which, after recursive application, allows one to express $\ket{D_{N}^m}$ in terms of the Dicke states of smaller number of particles
\begin{equation}\label{DickeDec}
\ket{D_{N}^m}=\sum_{i_1,\ldots,i_{N-m-1}=0}^1\frac{{{m+1 \choose m-\Sigma}}^{\frac{1}{2}}}{{{N\choose m}}^\frac{1}{2}}\ket{i_1,\ldots,i_{N-m-1}}\ket{D_{m+1}^{m-\Sigma}},
\end{equation}
where $\Sigma=i_1+\ldots+i_{N-m-1}$. Due to the fact that $\ket{D_{N}^m}$ is symmetric, the above decomposition holds true for any choice of $N-m-1$ parties in the first ket in (\ref{DickeDec}).

Having settled some basic information about the symmetric Dicke states, we can now move on to demonstrating how they can be self-tested. To facilitate our considerations we show how to self-test the following unitarlity equivalent state
\begin{eqnarray}\label{dupa3}
\ket{xD_{N}^m}&=&\sigma_z^{(N)}\ket{xD_{N}^m}\nonumber\\
&=&\sum_{i_1,\ldots,i_{N-m-1}=0}^1\frac{{{m+1 \choose m-\Sigma}}^{\frac{1}{2}}}{{{N\choose m}}^\frac{1}{2}}\ket{i_1,\ldots,i_{N-m-1}}\ket{xD_{m+1}^{m-\sum}}.
\end{eqnarray}
We then notice that the state corresponding to $i_1=\ldots=i_{N-m-1}=0$ is exactly $\ket{xD_{m+1}^m}=\sigma^{\otimes (m+1)}_x\ket{xW_{m+1}}$ with $\ket{xW_{m+1}}$ defined in (\ref{Wstate})). Moreover, since $\ket{xD_N^m}$ is symmetric on the first $N-1$ parties, the state  $\sigma^{\otimes (m+1)}_x\ket{xW_{m+1}}$ will appear in any decomposition of the form (\ref{dupa3}) in which any choice of $N-m-1$ parties from the first $N-1$ ones are in state $\ket{0}$. Importantly, we already know how to self-test the W state $\sigma^{\otimes (m+1)}_x\ket{xW_{m+1}}$. However, due to the transformation $\sigma_{x}^{\otimes (m+1)}$ we have to modify the conditions specified in Theorem \ref{ThmW} in the following way: 
\begin{eqnarray}\label{transf}
&&Z_i^{(+)}\leftrightarrow Z_i^{(-)}\quad (i=1,\ldots,N-1),\nonumber\\
&& D_N\to -E_N\quad \mathrm{and}\quad E_N\to -D_N. 
\end{eqnarray}

Now, to self-test a Dicke state $\ket{D_N^m}$ for any $m\geq \lfloor N/2\rfloor$ we can proceed in the following way:
\begin{enumerate}
\item Project any $(N-m-1)$-element subset $\mathcal{S}_i$ of the first $N-1$ parties of $\ket{\psi}$
(there are $\binom{N-1}{N-1-m}$ such subsets)
onto $\bigotimes_{j\in \mathcal{S}_i} Z_{j}^{(+)}$
and check whether the state corresponding to the remaining parties satisfies
the conditions for $\ket{xD_{m+1}^m}=\sigma^{\otimes (m+1)}_x\ket{xW_{m+1}}$.

\item For every sequence $\tau=(\tau_1,\ldots,\tau_N)$
consisting of $m+1$ ones on the first $N-1$ positions,  check that the state $\ket{\psi}$ obeys the following correlations 
\begin{equation}\label{VinaBujanda}
\langle\psi|Z_1^{(\tau_1)}\ot\ldots\ot Z_N^{(\tau_N)}|
\psi\rangle=0,
\end{equation}
where $Z_i^{(\tau_i)} = \left[\mathbb{1} + (-1)^{\tau_i}Z_i\right]/2$
\end{enumerate}
Let us now see in more details how the above procedure allows one to self-test $\ket{D_{N}^m}$. It is not difficult to see that the first condition leads us to the following decomposition
\begin{equation}\label{DickeIso}
\Phi_N(\ket{\psi}\ket{0}^{\otimes N})=\left[\bigotimes_{l\in\mathcal{S}_i}Z_{l}^{(+)}
\ket{\phi_i}\right]\otimes \left[\ket{0}^{\otimes (N-m-1)}_{\mathcal{S}_i}\ket{xD_{m+1}^m}\right]+\ket{\Phi_i}
\end{equation}
for any $i=1,\ldots,\binom{N-1}{N-1-m}$, where all $\mathcal{S}_i$ stand 
for different $(N-m-1)$-element subsets of the
$(N-1)$-element set $\{1,\ldots,N-1\}$, and $\ket{\phi_i}$ is defined as
\begin{equation}
\ket{\phi_i}=\bigotimes_{l\in\{1,\ldots,N\}\setminus\mathcal{S}_i}X_{l}Z_{l}^{(-)}\ket{\psi}.
\end{equation}
In other words, to construct $\ket{\phi_i}$ from $\ket{\psi}$
one has to act on the latter with $X_lZ^{(-)}_l$ on all
parties who do not belong to $\mathcal{S}_i$. Finally, $\ket{\Phi_i}$ 
is some state from the global Hilbert space collecting the remaining terms.

Let us now show that all the states
\begin{equation}
\ket{\widetilde{\phi}_i}=\bigotimes_{l\in\mathcal{S}_i}Z_{l}^{(+)}
\ket{\phi_i}
\end{equation}
are the same. To this end, we will exploit the conditions 
(\ref{W11}) and (\ref{W13}), which are clearly preserved under the transformation 
(\ref{transf}), and also the fact that [cf. (\ref{XX})]:
\begin{equation}\label{dupa13}
(X_i-X_N)\ket{\psi}=\ket{\psi}
\end{equation}
for any $i=1,\ldots,N-1$.
Consider two vectors $\ket{\widetilde{\phi}_i}$ and $\ket{\widetilde{\phi}_j}$ such that the corresponding sets $\mathcal{S}_i$ and $\mathcal{S}_j$ share $N-m-2$ elements (remember that these sets are equinumerous). Let $q$ and $p$ be the two elements by which these sets differ, i.e., $p\in\mathcal{S}_i$ ($q\in\mathcal{S}_j$) and $p\notin\mathcal{S}_j$ ($q\notin\mathcal{S}_i$). Then, using the condition (\ref{W13}) we turn the operator $X_tZ_t^{(-)}$ into $Z_t^{(+)}X_t$ at positions $t=q$ and $t=N$ for 
the state $\ket{\widetilde{\phi}_i}$, and, analogously, at positions $t=p$ and $t=N$ for the state $\ket{\widetilde{\phi}_j}$. We utilize the fact that $X_iX_N\ket{\psi}=\ket{\psi}$ for all $i=1,\ldots,N-1$ stemming from (\ref{dupa13}), which shows that $\ket{\widetilde{\phi}_i}=\ket{\widetilde{\phi}_j}$. Finally, repeating this procedure for all pairs of states for which the corresponding sets $\mathcal{S}_i$ differ by two elements, one finds that $\ket{\widetilde{\phi_i}}\equiv\ket{\phi}$ for all $i$.

As a result, the state (\ref{DickeIso}) simplifies to 
\begin{equation}
\Phi_N(\ket{\psi}\ket{0}^{\otimes N})=\ket{\phi} \ket{xD_N^m}+\ket{\Phi},
\end{equation}
$\ket{\Phi}$ is a vector from the global Hilbert space defined as
\begin{equation}\label{Queulat}
\ket{\Phi}=\sum_{\tau}\left(X_1^{\tau_1}Z_1^{(\tau_1)}\ot\ldots\ot X_N^{\tau_N}Z_N^{(\tau_N)}\ket{\psi}\right)\otimes\ket{\tau},
\end{equation}
where the summation is over all sequences $\tau=(\tau_1,\ldots,\tau_N)$ that contain less than $N-m-1$ zeros (or, equivalently, more than $m$ ones) on the first $N-1$ positions. 

Now, to prove that $\ket{\Phi} = 0$ it suffices to exploit the second step in the above procedure. That is, 
the condition (\ref{VinaBujanda}) is equivalent to 
\begin{equation}\label{vector}
Z_1^{(\tau_1)}\ot\ldots\ot
Z_N^{(\tau_N)}\ket{\psi}=0
\end{equation} 
for every sequence $(\tau_1,\ldots,\tau_N)$ consisting of $m+1$ ones at the first $N-1$ positions. Then, every component of the vector in (\ref{Queulat}) contains a sequence of at least $m+1$ $Z^{(-)}$ operators, which by virtue of (\ref{vector}) implies that $\ket{\Phi}=0$. This completes the proof.

For the self-testing of measurements the same argumentation as in the case of $W$-state self-test applies:
\begin{eqnarray*}
\Phi(Z_i\ket{\psi}\ket{0\dots 0}) 
= \ket{\phi}\otimes\sigma_z^{(i)}\ket{xD_N^{m}} \qquad (i=1,\dots, N-1),\\
\Phi(X_i\ket{\psi}\ket{0\dots 0}) 
= \ket{\phi}\otimes\sigma_x^{(i)}\ket{xD_N^{m}} \qquad  (i=1,\dots, N-1),\\
\Phi(Z_N\ket{\psi}\ket{0\dots 0})  =\ket{\phi}\otimes\left(\frac{\sigma_z^{(N)}+\sigma_x^{(N)}}{\sqrt{2}}\right)\ket{xD_N^{m}},\\
\Phi(X_N\ket{\psi}\ket{0\dots 0}) =\ket{\phi}\otimes\left(\frac{\sigma_z^{(N)}-\sigma_x^{(N)}}{\sqrt{2}}\right)\ket{xD_N^{m}}.
\end{eqnarray*}

Analogously to the $N$-partite $W$ state the total amount of correlators necessary 
for self-testing of any Dicke state scales linearly with $N$. This is because 
one essentially needs to obtain the same correlators as for the $W$ state.

\section{Self-testing the graph states}
\label{App:Graph}

In this section we provide the detailed proof of the self-test of graph states, stated in Theorem \ref{theorem3}.

Before proving the theorem we need to make some introductory remarks. Let $\ket{\psi_G}$ be an $N$-qubit graph state that corresponds to a graph $G=\{V,E\}$, where $V=\{1,\ldots,N\}$ and $E$ stand for the sets of vertices and edges, respectively. Recall that any graph state can be written in the following form
\begin{equation}\label{eq:GraphState}
\ket{\psi_G} = \frac{1}{\sqrt{2^N}}\sum_{\bi{i} \in \{0,1\}^N}(-1)^{\mu(\bi{i})}\ket{\bi{i}},
\end{equation}
where the summation is over all sequences $\bi{i}=(i_1,\ldots,i_N)$
with $i_j=0,1$, and $\mu(\bi{i})$ is the number of edges connecting qubits being in the state $\ket{1}$ in ket $\ket{\bi{i}}$ (without counting the same edge twice).

Let then $\nu_i$ be the set of neighbours of the $i$th qubit, that is, all those qubits that are connected with $i$ by an edge, while by $|\nu_i|$ we denote the number of elements in $\nu_i$. Likewise, we denote by $\nu_{i,j}$ the set of neighbours of a pair of qubits $i$ and $j$, i.e., all those qubits that are connected to either $i$ or $j$ (notice that $\nu_{i,j}=\nu_{j,i}$), and $|\nu_{i,j}|$ the number of elements of $\nu_{i,j}$. We also assume that the parties are labelled in such a way that qubits $N-1$ and $N$ are connected and the party $N$ has the smallest number of neighbours, i.e., $\vert \nu_N \vert \leq \vert \nu_i\vert$ for all $i$. 


The main property of the graph states underlying our simple self-test is that by measuring all the neighbours of a pair of connected qubits $i,j$ in the $\sigma_z$ basis, the two qubits $i$ and $j$ are left in one of the Bell states [cf. Prop. 1 in Ref. \cite{HeinPRA}]:
\begin{equation}
\frac{1}{\sqrt{2}}(\sigma_z^{m_i}\otimes \sigma_{z}^{m_j})
(\ket{0+}+\ket{1-})
\end{equation}
where $m_i$ ($m_j$) is the number of parties from set $\nu_{i,j} \setminus j$ ($\nu_{i,j} \setminus i$) 
whose result of a measurement in $\sigma_z$ basis was $-1$, and where we have neglected an unimportant 
$-1$ factor that might appear. 

%

We are now ready to prove the main theorem. Let us denote 
$Z^{(\tau)}_{\nu_{i,j}}=\otimes_{l\in\nu_{i,j}}Z_l^{(\tau_l)}$,
where $\tau$ is an $|\nu_{i,j}|$-element sequence with each $\tau_l\in \{+,-\}$ (the operator $Z^{(\tau)}_{\nu_{i,j}}$ acts only on the parties belonging to $\nu_{i,j}$).

\begin{theorem}Let $\ket{\psi}$ and measurements $Z_i,X_i$ with $i=1,\ldots,N-1$
and $Z_N,D_N,E_N,Z_N \equiv \frac{D_N-E_N}{\sqrt{2}},_N \equiv \frac{D_N+E_N}{\sqrt{2}}$ satisfy the following conditions
\begin{equation}\label{GraphCond1}
\left\langle Z^{(\tau)}_{\nu_{N-1,N}}\right\rangle=\frac{1}{2^{|\nu_{N-1,N}|}},\qquad
\left\langle Z^{(\tau)}_{\nu_{N-1,N}}\otimes B_{N-1,N}^{(m_{N-1},m_N)}\right\rangle=\frac{2\sqrt{2}}{2^{|\nu_{N-1,N}|}}
\end{equation} 
for every choice of the $|\nu_{i,j}|$-element sequence $\tau$. The Bell operators $B_{N-1,N}^{(m_{N-1},m_{N})}$ are defined as
\begin{equation}
B_{N-1,N}^{(m_{N-1},m_{N})}=(-1)^{m_N} X_{N-1}\otimes (D_N+E_N)+(-1)^{m_{N-1}}Z_{N-1}\otimes (D_N-E_N),
\end{equation}
where $m_{N-1}$ and $m_{N}$ are the numbers of neighbours of 
the qubits, respectively, $N-1$ and $N$ (excluding the $N$th qubit and $N-1$th qubit, respectively) which are projected onto the eigenvector of $Z_i^-$.

We then assume that 
\begin{eqnarray}\label{GraphCond2}
\left\langle Z^{(\tau)}_{\nu_{i,j}}\right\rangle=\frac{1}{2^{|\nu_{i,j}|}},\qquad
\left\langle Z^{(\tau)}_{\nu_{i,j}}\otimes Z_i\otimes X_j\right\rangle=\frac{(-1)^{m_{j}}}{2^{|\nu_{i,j}|}}
\end{eqnarray}
for all connected pairs of indices $i\neq j$ except for $\neq (N,N-1)$ . As before, 
$Z^{(\tau)}_{\nu_{i,j}}=\otimes_{l\in\nu_{i,j}} Z_{l}^{(\tau_l)}$.
%
Then, for the isometry $\Phi_N$ one has
\begin{equation}\label{Wfinal}
\Phi_N\left(\ket{\psi}\ket{0}^{\otimes N}\right)=\ket{\mathrm{aux}}\ket{\psi_G}.
\end{equation}
\end{theorem}
\begin{proof}
The conditions in (\ref{GraphCond1}) imply that the 
normalized state 
\begin{equation}
\ket{\widetilde{\psi}^{(\tau)}_{N-1,N}}=\sqrt{2^{|\nu_{N-1,N}|}}\,Z_{\nu_{N-1,N}}^{(\tau)}\ket{\psi}
\end{equation}
violates maximally the CHSH Bell inequality, 
which in turn implies that 
\begin{equation}
\{X_{N-1},Z_{N-1}\}\ket{\widetilde{\psi}^{(\tau)}_{N-1,N}}=0
\quad\mathrm{and}\quad
\{X_{N},Z_{N}\}\ket{\widetilde{\psi}^{(\tau)}_{N-1,N}}=0,
\end{equation}
where $X_N=(D_N+E_N)/\sqrt{2}$ and $Z_N=(D_N-E_N)/\sqrt{2}$.
These identities hold true for any of $2^{|\nu_{i,j}|}$ projected states 
$\ket{\widetilde{\psi}^{(\tau)}_{N-1,N}}$, and therefore 
it must also hold for the initial state $\ket{\psi}$, i.e., 
\begin{equation}\label{DupaBlada}
\{X_{N-1},Z_{N-1}\}\ket{\psi}=0\quad
\mathrm{and}\quad \{X_{N},Z_{N}\}\ket{\psi}=0.
\end{equation}
This is because one can always decompose $\ket{\psi}$ in the eigenbasis of the operator $Z_{\nu_{N-1,N}}$ which is a tensor product of $Z$ operators acting on the neighbours of $i,j$.

Then, let us focus on the second bunch of conditions (\ref{GraphCond2}). They imply that the 
length of the projected vectors $\ket{\psi_{i,j}^{(\tau)}}=
Z_{\nu_{i,j}}^{(\tau)}\ket{\psi}$ is $1/\sqrt{2^{|\nu_{i,j}|}}$, so is 
the norm of $Z_i\ket{\psi_{i,j}^{(\tau)}}$ and $X_j\ket{\psi_{i,j}^{(\tau)}}$ for any connected pair $i\neq j$. This together with (\ref{GraphCond2}) mean that the vectors 
$Z_i\ket{\psi_{i,j}^{(\tau)}}$ and $X_j\ket{\psi_{i,j}^{(\tau)}}$
are parallel or antiparallel, or, more precisely, 
that
\begin{equation}\label{ThmGraphRownania}
(-1)^{m_j}Z_i\ket{\psi_{i,j}^{(\tau)}}=X_j\ket{\psi_{i,j}^{(\tau)}}
\end{equation}
for any connected pair of parties $i\neq j$.

%
%
%
%

Let us now consider one of the parties connected to 
the party $N-1$ (there must be at least one such party as otherwise 
the $N$th one would not be the one with the smallest number of neighbours or the graph would be bipartite).  We label this vertex by $N-2$. It then follows from conditions (\ref{ThmGraphRownania}) that for the particular pair of vertices $N-2,N-1$, one has the following identities  
\begin{equation}\label{ThmGraph1}
X_{N-2}\ket{\psi_{N-2,N-1}^{(\tau)}}=(-1)^{m_{N-2}}Z_{N-1}\ket{\psi_{N-2,N-1}^{(\tau)}}
\end{equation}
and
\begin{equation}\label{ThmGraph2} 
Z_{N-2}\ket{\psi_{N-2,N-1}^{(\tau)}}=(-1)^{m_{N-1}}X_{N-1}\ket{\psi_{N-2,N-1}^{(\tau)}}
\end{equation}  
hold true for all configurations of $\tau$. With their aid the 
following sequence of equalities
\begin{eqnarray}\label{ThmGraph3}
X_{N-2}Z_{N-2}\ket{\psi_{N-2,N-1}^{(\tau)}}&=&(-1)^{m_{N-1}}X_{N-2}X_{N-1}\ket{\psi_{N-2,N-1}^{(\tau)}}\nonumber\\
&=&(-1)^{m_{N-1}+m_{N-2}}X_{N-1}Z_{N-1}\ket{\psi_{N-2,N-1}^{(\tau)}}\nonumber\\
&=&-(-1)^{m_{N-1}+m_{N-2}}Z_{N-1}X_{N-1}\ket{\psi_{N-2,N-1}^{(\tau)}}\nonumber\\
&=&-(-1)^{m_{N-2}}Z_{N-1}Z_{N-2}\ket{\psi_{N-2,N-1}^{(\tau)}}\nonumber\\
&=&-Z_{N-2}X_{N-2}\ket{\psi_{N-2,N-1}^{(\tau)}}
\end{eqnarray}
holds true for any choice of $\tau$, where first and the second equality stems from (\ref{ThmGraph2}) and (\ref{ThmGraph1}), respectively, the third one is a consequence of the anticommutativity of $X_{N-1}$ and $Z_{N-1}$. Finally, the 
fourth and the fifth equality follows again from (\ref{ThmGraph2}) and (\ref{ThmGraph1}), respectively. 

Since the identity (\ref{ThmGraph3}) is obeyed 
for any configuration of $\tau$, it must also hold for the 
state $\ket{\psi}$, that is,  
%
$\{X_{N-2},Z_{N-2}\}\ket{\psi}=0.$
%
Taking into account the conditions (\ref{GraphCond1}), this procedure can be recursively applied to any pair of 
connected particles, yielding (together with (\ref{DupaBlada}))
\begin{equation}
\label{eq:acr}
\{X_{i},Z_{i}\}\ket{\psi}=0
\end{equation}
for any $i=1,\ldots,N$.
\end{proof}

The action of the isometry is given by 
\begin{equation}\label{GraphIsometry}
\Phi_N(\ket{\psi}\ket{0}^{\otimes N})=\sum_{i_1,\ldots,i_N=0}^{1}
X_{1}^{i_1}\ldots X_N^{i_N}\,
Z_{1}^{(i_1)}\ldots Z_N^{(i_N)}\,\ket{\psi}\ket{i_1\ldots i_N}
\end{equation}

Let us now consider a particular term in the above decomposition in which the sequence 
$i_1,\ldots,i_N$ has $k>0$ ones at positions $j_1,\ldots,j_k$, i.e., 
\begin{equation}\label{GraphTerm}
X_{j_1}^{i_{j_1}}\ldots X_{j_k}^{i_{j_k}}\,
\bigotimes_{l\notin I}Z_l^{(+)}\,
\bigotimes_{l\in I}Z_l^{(-)}\ket{\psi},
\end{equation}
where $I=\{j_1,\ldots,j_k\}$. By using the previously derived relations, we want 
to turn this expression into one that is proportional to the auxiliary state 
$Z_{1}^{(+)}\otimes\ldots \otimes Z_N^{(+)}\ket{\psi}$. To this end, let us first 
focus on the party $j_k$ and consider one of its neighbours
which we denote by $l$. For this pair of parties, the conditions 
(\ref{ThmGraphRownania}) imply that 
\begin{equation}
X_{j_k}\ket{\psi_{{j_k,l}}}=(-1)^{m_{j_k}}Z_{l}\ket{\psi_{{j_k,l}}},
\end{equation}
where, to recall, $m_{j_k}$ is the number of neighbours of $j_k$ being in the state $\ket{1}$ except for $l$. The above identity together with the anticommutativity relation 
$\{X_{j_k},Z_{j_k}\}\ket{\psi}=0$ allow us to replace in (\ref{GraphTerm}) the 
operator $X_{j_k}^{i_k}Z_{j_k}^{(-)}$ by $(-1)^{m_{j_k}}Z_{j_k}^{(+)}Z_{l}$.
Now, if the value of $i_l$ in the corresponding ket $\ket{i_1,\ldots,i_N}$
is zero, the last operator $Z_l$ can be simply absorbed by 
$Z_l^{(+)}$, while if $i_l=1$, one uses that fact that $Z_l^{(-)}Z_l=-Z_l^{(-)}$, meaning that one has an additional minus sign. Altogether this turns the operator 
$X_{j_k}^{i_k}Z_{j_k}^{(-)}$ into $(-1)^{m_{j_k}'}Z_{j_k}^{(+)}$, where 
$m_{j_k}'$ is the number of neighbours of $j_k$ (including $i_l$) 
which are in the state $\ket{1}$. Plugging this into (\ref{GraphTerm}) we can rewrite the latter as
\begin{equation}\label{GraphTerm2}
(-1)^{m_{j_k}'}X_{j_1}^{i_{j_1}}\ldots X_{j_{k-1}}^{i_{j_{k-1}}}\,
\bigotimes_{l\notin I'}Z_l^{(+)}\,
\bigotimes_{l\in I'}Z_l^{(-)}\ket{\psi},
\end{equation}
where now $I'=I\setminus\{i_l\}$, and so we have lowered the number of
elements of $I$ by one. It should be stressed that this affects the numbers of
neighbours of those parties that are still in $I'$, which will be of importance 
for what follows. 

Now, we can apply recursively the same reasoning to the remaining
elements of $I'$, keeping in mind that at each step one element
is removed from $I'$. We thus arrive at 
\begin{equation}
(-1)^{\mu'(\bi{i})}\,Z_1^{(+)}\,
\otimes \ldots \otimes Z_N^{(+)}\ket{\psi},
\end{equation}
with $\mu'$ defined as
\begin{equation}
\mu'(\bi{i})=\sum_{l=1}^k m_{j_l}^{>}, 
\end{equation}
where $m_{j_l}^{>}$ is the number of neighbours of $i_{j_l}$ being in the state 
$\ket{1}$ and having smaller indices than $j_l$, or, in other words, 
those elements of $I=\{j_1,\ldots,j_k\}$ smaller than $j_l$
that are neighbours of $i_{j_l}$. One immediately notices that 
$\mu'(\bi{i})$ equals $\mu(\bi{i})$ for a given $\bi{i}$, and therefore 
by applying the above reasoning to every term in 
(\ref{GraphIsometry}), one arrives at
\begin{equation}
\Phi_N(\ket{\psi}\ket{0}^{\otimes N})=
\left(Z_{1}^{(+)}\otimes\ldots \otimes Z_N^{(+)}\ket{\psi}\right)
\otimes\sum_{\bi{i}}(-1)^{\mu(\bi{i})}\ket{\bi{i}},
\end{equation}
which after normalizing both terms can be written as
\begin{equation}
\Phi_N(\ket{\psi}\ket{0}^{\otimes N})=
\ket{\mathrm{aux}}\ket{\psi_{G}}
\end{equation}
with $\ket{\mathrm{aux}}=(1/\sqrt{2^N})(Z_{1}^{(+)}\otimes\ldots \otimes Z_N^{(+)}\ket{\psi})$. \\

Once relation (\ref{eq:acr}) is satisfied for all $i$s the proof for measurement self-testing goes along the same lines as the proof for the self-testing of the state. Let us check how isometry $\Phi_n$ acts on the state $X_{\tilde{i}}\ket{\psi}$. (\ref{GraphTerm2}) takes the form:
\begin{eqnarray*}
\Phi_N(X_{\tilde{i}}\ket{\psi}\ket{0}^{\otimes N}) &=& \sum_{I,l}(-1)^{m_{j_k}'}X_{j_1}^{i_{j_1}}\ldots X_{j_{k-1}}^{i_{j_{k-1}}}\,
\bigotimes_{l\notin I'}Z_l^{(+)}\,
\bigotimes_{l\in I'}Z_l^{(-)}X_{\tilde{i}}\ket{\psi}\\
&=& \sum_{I \oplus \tilde{i},l}(-1)^{m_{j_k}'+1}X_{j_1}^{i_{j_1}}\ldots X_{j_{k-1}}^{i_{j_{k-1}}}\,
\bigotimes_{l\notin I'}Z_l^{(+)}\,
\bigotimes_{l\in I'}Z_l^{(-)}\ket{\psi}\\
&=& \left(Z_{1}^{(+)}\otimes\ldots \otimes Z_N^{(+)}\ket{\psi}\right)
\otimes\sum_{\bi{i}}(-1)^{\mu(\bi{i}\oplus \tilde{i})}\ket{\bi{i}}\\
&=& \ket{\mathrm{aux}}\otimes{\sigma_x}^{(\tilde{i})}\ket{\psi_{G}}, \quad \forall \tilde{i} \in \{1,2,\dots, N-1\}
\end{eqnarray*}
where $I \oplus \tilde{i}$ is equal to $I / \tilde{i}$ if $\tilde{i} \in I$ and to $I \cup \tilde{i}$ otherwise, and $\mu(\bi{i}\oplus \tilde{i})$ is the number of edges connecting qubits being in the state $\ket{1}$ in ket $\ket{\bi{i} \oplus (0\dots 0\tilde{i}0\dots 0)}$ (without counting the same edge twice). Similarly it can be shown that
\begin{eqnarray*}
\Phi_N(Z_{\tilde{i}}\ket{\psi}\ket{0}^{\otimes N}) &=&  \ket{\mathrm{aux}}\otimes {\sigma_z}^{(\tilde{i})}\ket{\psi_{G}},\quad \forall \tilde{i} \in \{1,2,\dots, N\},\\
\Phi_N(D_{N}\ket{\psi}\ket{0}^{\otimes N}) &=&  \ket{\mathrm{aux}}\otimes\left(\frac{\sigma_x^{(N)}+\sigma_z^{(N)}}{\sqrt{2}}\right)\ket{\psi_{G}},\\
\Phi_N(E_{N}\ket{\psi}\ket{0}^{\otimes N}) &=&  \ket{\mathrm{aux}}\otimes\left(\frac{-\sigma_x^{(N)}+\sigma_z^{(N)}}{\sqrt{2}}\right)\ket{\psi_{G}}.
\end{eqnarray*}

Finally, for self-testing graph-states one has to measure $3+\vert E \vert$ correlators, where $\vert E \vert$ is the total number of edges, which even for the fully connected graph is strictly better scaling than the complexity of quantum state tomography.


\end{document}